\documentclass[acmsmall]{ec19acm}

\AtBeginDocument{%
  \providecommand\BibTeX{{%
    \normalfont B\kern-0.5em{\scshape i\kern-0.25em b}\kern-0.8em\TeX}}}

\copyrightyear{2019} 
\acmYear{2019} 
\setcopyright{acmlicensed}
\acmConference[EC '19]{ACM EC '19: ACM Conference on Economics and Computation}{June 24--28, 2019}{Phoenix, AZ, USA}
\acmBooktitle{ACM EC '19: ACM Conference on Economics and Computation (EC '19), June 24--28, 2019, Phoenix, AZ, USA}
\acmPrice{15.00}
\acmDOI{10.1145/3328526.3329581}
\acmISBN{978-1-4503-6792-9/19/06}

\usepackage{amsmath,amssymb}
\usepackage[mathscr]{eucal}
\usepackage{algorithmic}
\usepackage{mathtools}
\usepackage{bbm}

\newcommand{\E}{\mathbb{E}}

\newenvironment{proof*}[1]
  {\renewcommand{\proofname}{#1}%
  \begin{proof}}
  {\end{proof}}

\begin{document}

\title{Optimal Algorithm for Bayesian Incentive-Compatible Exploration}

\author{Lee Cohen}
\authornote{This work was supported in part by The Yandex Initiative for Machine Learning.}
\email{leecohencs@gmail.com}
\orcid{0000-0002-3153-7965}
\affiliation{%
  \institution{Blavatnik School of Computer Science, Tel Aviv University}
    \country{Israel}
}

\author{Yishay Mansour}
\authornote{This work was supported in part by a grant from ISF.}
\affiliation{%
  \institution{Blavatnik School of Computer Science, Tel Aviv University \& Google Research}
  \country{Israel}}
\email{mansour.yishay@gmail.com}
\orcid{0000-0001-6891-2645}

\renewcommand{\shortauthors}{Cohen and Mansour}

\begin{abstract}
We consider a social planner faced with a stream of myopic selfish agents. The goal of the social planner is to maximize the social welfare, however, it is limited to using only information asymmetry (regarding previous outcomes) and cannot use any monetary incentives. The planner recommends actions to agents, but her recommendations need to be \textit{Bayesian Incentive Compatible} to be followed by the agents.

Our main result is an {\bf optimal} algorithm for the planner, in the case that the actions realizations are deterministic and have limited support, making significant important progress on this open problem. Our optimal protocol has two interesting features. First, it always completes the exploration of {\em a priori} more beneficial actions before exploring a priori less beneficial actions. Second, the randomization in the protocol is correlated across agents and actions (and not independent at each decision time).
\end{abstract}



\maketitle

\section{Introduction}
The inherent trade-off between exploration and exploitation is at the core of any reactive learning algorithm.
Multi-arm bandit is a simple model which highlights this inherent trade-off.
Multi-arm bandits can model a variety of scenarios, including pricing (where the actions are prices), recommendation (e.g., where actions are news articles) and many other settings.
To a large part, multi-arm bandit is viewed as a model for learning and optimization in which the planner can select any available action.
However, when we are considering human agents as the entities performing the action, then incentives become a major issue. While a planner can recommend actions to the agents (in order to explore different alternatives), the agents ultimately decide whether to follow the given recommendation. This raises the issue of incentives in addition to the exploration-exploitation trade-off.
The planner can induce explorations in many ways. The simplest is using monetary transfers, paying the agents in order to explore (for example, Frazier et al. \cite{FKKK14}).
We are interested in the case when the social planner is unable or prefers to avoid any monetary transfers.
(This can be due to regulatory constraints, business model, social norms, or any other reason.)
The main advantage of the planner in our model is the \textit{information asymmetry}, namely, the fact that the planner has much more information than the agents.
As a motivating example for \textit{information asymmetry}, consider a GPS driving application. The application (social planner) is recommending to the drivers (agents) the best route to drive (action), given the changing road delays, and observes the actual road delays when the route is driven. While the application can recommend driving routes, ultimately, the driver decides which route to actually drive. The application needs periodically to send drivers on exploratory routes, where it has uncertainty regarding the actual delay, in order to observe their delay. The driver is aware that the application has updated information regarding the current delays on various roads. For this reason, the driver would be willing to follow the recommendation even if she knows that there is a small probability that she is asked to explore. On the other extreme, if the driver would assume that with high probability a certain recommended route has a higher delay, she might drive an alternate route. This inherent balancing of exploration and exploitation while satisfying agents' incentives, is at the core of this work.


The abstract model that we consider is the following. 
There is a finite set of actions, and for each action there is a prior distribution on its rewards.
A social planner is faced with a sequence of myopic selfish agents, and each agent appears only once. The social planner would like to maximize the social welfare, the sum of the agents' utilities. The social planner recommends to each agent an action, and if the recommendation is \textit{Bayesian incentive compatible} (henceforth, \textit{BIC}), the agent will follow the action. This model was presented in Kremer et al. \cite{Kremer-JPE14} and studied in \cite{MansourSS15,MansourSSW16,MansourSW18}. The work of Kremer et al. \cite{Kremer-JPE14} presented an optimal algorithm for the social planner in the case of {\em two} actions with deterministic outcome. (Deterministic outcome implies that each time the action is performed we receive the same reward, and the uncertainty is what that value will be, which is governed by the prior distribution.)

Our main focus is to make progress on this important open problem of providing an \textbf{optimal} policy for this setting for any number of actions. For this end, we consider a somewhat more restricted setting, where each action has a finite support. 
If we assume that there are only two possible values, say $\{-1,+1\}$, then the task becomes trivial.
We can simply order the agents according to the actions' expectation, and ask them to explore until we reach an action of value $+1$, and then recommend it forever. This would work even if we provide the agents with the realizations of the previous actions. In this work we take a small, yet significant, step away from this trivial model. We assume that the best a priori action has a larger support. For the most part we analyze the case that the support of the a priori best action is $\{-1,0,+1\}$, while the other actions have support $\{-1,+1\}$. We later extend our results to handle a more general setting of any continuous distribution with full support on $[-1,1]$ for the a priori best action (in Section \ref{contSec}).

Our simple model has a significant complexity and allows us to draw a few interesting insights. To understand the challenges, consider the case where the actions have a negative expected reward. (For simplicity, we assume that the actions are sorted by their expected reward, where action $1$ has the highest expectation.) In such a case, if the realization of action $1$ is $+1$, clearly the planner would recommend it for all the following agents. If the realization of action $1$ is $-1$, clearly any other action is superior to it. However, the challenging case occurs when the realization of action $1$ is $0$. In this case, the selfish agents would prefer to perform action $1$ with $0$ reward (since other actions have negative expected reward). The challenge to the social planner is to incentivize the agents to explore. The main idea is that of \textit{information asymmetry}. When the planner recommends action $2$, the agent is unsure whether the social planner observed that outcome of action $1$ is $-1$, in which she would like to perform it, or whether the social planner observed that the outcome is $0$ and asks the agent to explore. The social planner, by a delicate balancing of the exploration probability, can make the recommendation BIC.


Our main result is an optimal algorithm for the social planner when faced with $k$ actions, both for support $\{-1,0,+1\}$ and $[-1,+1]$ for the best apriori action. First, the algorithm makes sure that the BIC constraints are tight, which is a simple intuitive requirement and is clearly required for optimality.
However, we need to exhibit much more refined properties to construct an optimal algorithm. 
%
An interesting issue regarding the exploration order is whether when we force a tight BIC constraint we might be forced to explore an action $j$ before we know the values of actions $1, \ldots , j-1$ (that have better expected reward than action $j$).
We show that this is not the case in the optimal algorithm, namely,
the exploration of action $j$ starts only after the social planner knows the realizations of all the better a priori actions, i.e., $1, \ldots, j-1$. 
While this seems like an intuitive outcome, it relies on the very delicate way in which our algorithm performs its randomization.
(Recall that the recommendation algorithm uses randomization to balance between exploring and exploiting.)

The implementation of the randomization is the second interesting property of our algorithm.
In our randomization, we use a correlation between agents and actions. Specifically, the randomization selects for each action a random agent that might potentially explore it (if needed). Special care needs to be taken to make sure that for different unexplored actions we always select different agents. 

We show that our algorithm does not only maximize the social welfare but in addition minimize the exploration time, the time until the social planner does not need to explore any more. For the most part we assume that the number of agents is large enough that the social planner completes the exploration. We show also how to derive the optimal policy in the case of a limited number of agents. 


\subsection*{Related works}

As mentioned, the work of Kremer et al. \cite{Kremer-JPE14} presented the model and derived the optimal policy for two deterministic actions.
Mansour et al. \cite{MansourSS15} derive tight asymptotic regret bound in the case of stochastic actions as well as a reduction
from an arbitrary non-BIC policy to a BIC one. 
Bahar et al. \cite{BaharST16}
enrich the model by embedding the agents in a
social network, and allowing them to observe their neighbors.
Mansour et al. \cite{MansourSSW16} extended the model to allow a multi-agent game in each time step, rather than a single agent.
Mansour et al. \cite{MansourSW18} consider the case of two competing planners. 

Frazier et al. \cite{FKKK14} consider a model with monetary transfers, where the social planner can pay agents to explore. 
Che and H{\"o}rner \cite{Che-13} consider a setting with two binary-valued actions and continuous information flow
and a continuum of agents.
Finally, Slivkins \cite{Slivkins17} has an excellent overview of the topic.

%

Another related topic is that of \textit{Bayesian Persuasion} by Kamenica and Gentzkow \cite{Kamenica-aer11} where the planner tries to infer a value of an ``unobservable'' state using interaction 
with multiple agents. See \cite{DughmiKQ16,DughmiX17,DughmiX16} for a more
algorithmic perspective of Bayesian Persuasion.

Multi-armed bandit \cite{regret-book,Gittins-book11} is a
well-studied model for exploration-exploitation trade-off both in operations research and machine learning.
The main focus in learning multi-arm bandits is on designing efficient algorithms that have a guaranteed performance compared to the best single action.
\section{Model}
Let $A:=\{1,2,...,k\}$ be the set of possible actions.
The prior distribution $D=D_1\times D_2\times ...\times D_k$ defines random variables $X_j$ for the rewards of actions $j\in A$. The reward of action $j\in A$, denoted by $x_j$, is sampled from $D_j$ (it is sampled once, and any application of action $j$ yields the same reward $x_j$).
The prior expected reward of action $j$ is $\mu_j:=\E_{D_j}[X_j]$, and for notational convenience we assume that $\mu_1>\mu_2>\cdots >\mu_k$.

In this work we focus on the case that the support of distribution of $D_1$ is $\{-1,0,+1\}$ (the case of support $[-1,1]$ appears in Section \ref{contSec}). The support of distribution $D_j$, for $j\geq 2$, is $\{+1,-1\}$. We denote by $p_j^\alpha:=\Pr[X_j=\alpha]$, which implies that the distribution $D_j$, for $j\geq 2$, has a single parameter, $p_j^1$ (and $p_j^{-1}=1-p_j^1$). 
We assume that $p^{1}_j>0$, otherwise the action has a constant reward of $-1$.

The interaction between the planner and the agents proceeds as follows.
At time $t$, the $t$-th agent arrives, and the planner recommends to the $t$-th agent action $\sigma_t \in A$, which is called {\em the recommended action}. 
Given the recommended action $\sigma_t$, the $t$-th agent selects an action $a_t$, receives a reward $x_{a_t}$, and leaves. Formally, the $t$-th agent has a \textit{utility function}, $u_t$, and $u_t(a)=x_{a}$ if action $a$ has been explored, else $\E[u_t(a)]=\mu_{a}$.
A {\em history} at time $t$, $h_{t}$, contains all the previous chosen actions by the agents, i.e., $a_1,\ldots,a_{t}$, and their corresponding rewards, $x_{a_1},\ldots,x_{a_t}$.
A strategy for the planner is a {\em recommendation policy}, $\pi$, where $\pi_t(h_{t-1})=v_t\in \Delta(A)$, where $\Delta(A)$ is the set of distributions over $A$, i.e.,
$\Delta(A)=\{v_t\in \mathbb{R}^k|\forall j\in A, v_t[j]\geq0\ and\ \Sigma_{j=1}^k v_t[j]=1\}$.
The value of $v_t[j]$ is the probability that $\sigma_t=j$, i.e., $v_t[j]=\Pr[\sigma_t=j]$.

A recommended action, $\sigma_t$, is \textit{Bayesian incentive-compatible (BIC)} if for any action $j\in A$, we have $\E[u_t(\sigma_t)-u_t(j)|D,\pi,\sigma_t,t]\geq 0$. \footnote{The expectation is implicitly conditioned on the actions that were selected by the previous agents, as it conditioned on the policy, the prior and the agent's place in line.} Such constrains are called \textit{BIC constrains}.
I.e., if $\sigma_t$ is BIC then there is no other action $j\in A$ that can increase agent $t$'s expected reward, based on the prior $D$, the policy $\pi$, the recommended action $\sigma_t$, and the agent's place in line $t$, all of which are known to agent $t$ before selecting an action (note that the agents do not observe the history $h_{t-1}$). 
A recommendation policy for the planner, $\pi$, is BIC if all it's recommendations are BIC. Namely, for any agent $t$ and any history with positive probability $h_{t-1}$, the recommendation $\sigma_t$ is BIC. 

The {\em social welfare} is the expected cumulative reward of all the agents. The social welfare of a BIC recommendation policy $\pi$ is:
$SW_T(\pi):=\E[\Sigma_{t=1}^T u_t(\sigma_t)]=\E[\Sigma_{t=1}^T u_t(\pi_t(h_{t-1}))]$.

The Bayesian prior $D$ on the rewards, is a common knowledge to the planner as well as all the agents. W.l.o.g, we restrict the planner's recommendation policy to be BIC, which assures that the agents follow the recommended actions.
Our main goal is to design a BIC algorithm that maximizes social welfare (i.e., the cumulative reward of the agents). 
\section{Optimal BIC Algorithm for $k$ actions}

We start with a simpler case that will have most of the ingredients of the more general case. We restrict the first action to have only three possible values $\{-1,0,1\}$, namely, the support of $D_1$ is $\{-1,0,1\}$. The second restriction is that we assume that there are only three actions, i.e., $k=3$. The terminology is provided for $k$-actions settings, but some of the intuition and motivation 
are provided for three actions settings. 
The proofs appear in Appendix~\ref{app:sec3}. The algorithm
for the general case of $k>3 $ actions, and some of its proofs are in the appendix \ref{appendix1}.

Given this special case, we claim that the challenging case is when $0>\mu_2>\mu_3$. In the case that $\mu_1>\mu_2>\mu_3>0$, we can simply recommend to the first agent action $1$, i.e., $\sigma_1=1$. When we observe $x_1$, then: (1) If $x_1=1$, we recommend to all the agents action $1$, i.e., $\sigma_t=1$. (2) If $x_1=0$ or $x_1=-1$, we recommend to the second agent action $2$, i.e., $\sigma_2=2$. This is BIC since $\mu_2>0\geq  x_1$ in this case. If $x_2=1$ we recommend to all the agents $\sigma_t=2$. Otherwise, $x_2=-1$, and we recommend to the third agent action $3$, i.e., $\sigma_3=3$. Again, this is BIC since $\mu_3>0\geq x_1 \geq x_2$. 
Either way, all the agents after the first three will be performing the optimal action. The above policy maximizes social welfare even if we do not restrict the information flow, and the planner announces to the agents the actions' realizations.
In the case that $\mu_1>\mu_2>0>\mu_3$, we can execute for the first two agents the above strategy, and essentially reduce the number of actions to two, for which the optimal policy was given by Kremer et al.  \cite{Kremer-JPE14}.
For this reason, we assume that $0>\mu_2>\mu_3$. (And for $k$ actions, we assume $0>\mu_2>\dots >\mu_k$.)

To build intuition we start with a simple example, in order to explain how a BIC policy can give a recommendation $\sigma_t\ne 1$.
\begin{example}\label{exampleSigmat2}
Consider a recommendation $\sigma_t=2$ to agent $t$. The possible reasons for it is one of the following:
\begin{enumerate}
\item \textbf{Exploitation driven recommendation}: Action $2$ is the best action given the history. This can be due to one of the following cases:
\begin{enumerate}
\item \textbf{A known reward}: The planner already observed that $x_{2}=1$, which is the maximum possible reward. From that time, the recommended action is $\sigma_t=2$, as it has the maximum possible reward.
\item \textbf{An unknown reward}: The observed realizations have the minimum possible reward, i.e., $x_1=- 1$ and maybe $x_3=-1$. Given this realization, we know that $\E[u_t(2)]=\mu_2>x_1$ (and in case that $x_3=-1$, also $\E[u_t(2)]=\mu_2\geq x_3$). This makes action $2$ the best action to execute, considering the history.
\end{enumerate}
\item \textbf{Exploration driven recommendation}: The planner has not yet observed an action with the best possible reward (i.e., $1$), and observed $x_1=0$. Since we assume that $0>\mu_2>\mu_3$, such a recommendation would not benefit for agent $t$ (but the planner is recommending it since it might benefit future agents).
\end{enumerate}
\end{example}
Fortunately, the agents do not know the realizations of the actions' rewards, hence cannot infer the reason for their recommendations. This is where the \textit{information asymmetry} translates into an advantage for the planner, and enables her to maximize social welfare. 

\subsection{Information States}

It would be very useful to partition the histories depending on the information that the planner has, regarding the realized values of the actions.
Since we have only three actions, we have at most three realized values, and we can encode them in a vector of length three. We use the $*$ symbol to indicate that a value is still unknown. For example, 
$\langle 0,-1,*\rangle$ implies that we know that $x_1=0$, $x_2=-1$ and we never explored the value of $X_3$.
Any history of the first $t-1$ agents which is compatible with $\vec{z}=\langle 0,-1,*\rangle$
is assigned to the information state $S^{\vec{z}}_t$. The recommendation to the $t$-th agent would depend on the planner's information state.

Note that the agents do not know the planner's information state.
However, given the recommendation $\sigma_t$, and the planner policy $\pi$, they can deduce the probabilities of each state, conditioned on the recommendation $\sigma_t$ they received. Those probabilities allow them to test whether the recommended action is indeed BIC, i.e., maximizes their expected reward given the information they observe.

Going back to example \ref{exampleSigmat2}, we can now describe it using information states.
\begin{example}\label{expstates}
Consider a recommendation to agent $t$,  $\sigma_t=2$. Every possible reason for it can be one of the following: 
\begin{enumerate}
\item \textbf{States that result in exploitation driven recommendation}, action $2$ either has:
\begin{enumerate}
\item \textbf{A known reward}: The planner has already observed action $2$'s reward and it is the maximum possible reward. I.e., the planner is in one of the following information states: 
$S_t^{\langle - 1,1,*\rangle}$, $S_t^{\langle- 1,1,- 1\rangle}$, $S_t^{\langle 0,1,-1\rangle}$ or $S_t^{\langle 0,1,*\rangle}$.
\item \textbf{An unknown reward}: The only action with a better prior expected reward compared to action $2$, action $1$, has been explored and resulted in minimal reward (i.e, $x_1=-1$). Action $2$ that now has the best utility, has not yet explored. We denote this state with $S_t^{\langle-1,*,*\rangle}$. (An additional possible state is $S_t^{\langle-1,*,-1\rangle}$ where the planner also observed that $x_3=-1$.)
\end{enumerate}
The set of these \textit{exploitation states} is denoted by $\Gamma^{j+}_t$, for the reason that following a recommendation for action $j$ in such states produces higher expected utility for agent $t$ compared to action $1$.
\item \textbf{States that may result in exploration driven recommendation}: Action $2$ has not been explored yet, whereas $x_1=0$. This implies that the planner is either in information state $S_t^{\langle 0,*,*\rangle}$,  or in information state $S_t^{\langle 0,*,-1\rangle}$. \\
The set of these \textit{exploration states} is denoted by $\Gamma^{j-}_t$, for the reason that following a recommendation for action $j$ in such states produces lower expected utility for agent $t$ then selecting action $1$.
\end{enumerate}
\end{example}

\subsection{The optimal BIC recommendation algorithm}

Given the information states, we can describe the planner's recommendation policy.
The recommendation policy will map the information states to recommended actions.
In the case of an ``Exploration driven Recommendation'' the mapping would be stochastic, to make sure that the incentives are maintained.
{\tt Algorithm 3-actions} is described in Table ~\ref{RecTable}, defining what recommendation to give in each information state.

{\tt Algorithm 3-actions} uses two functions, $f^2_t(y)$ and $f^3_t(y)$, which control the exploration and are based on a mutual parameter $y$, which will be selected uniformly at random in $[0,1]$. The states are marked also as {\em terminal states} if there is a unique recommendation for all future agents, and {\em exploration} if the recommended action might not have the highest expected reward ($S_t^{\vec{z}}\in \Gamma_t^{j-}$). States not marked as \textit{exploration} result in a exploitation driven recommendation, and are therefore \textit{exploitation} states ($S_t^{\vec{z}}\in \Gamma_t^{j+}$).

Looking at {\tt Algorithm 3-actions} in Table ~\ref{RecTable} might be intimidating, however, in most of the information states the recommendations are rather straightforward. In the initial information state, i.e., $\langle *,*,*\rangle$, the only BIC recommendation is action $1$, since the first agent knows that the planner has no additional information beyond the prior. In any information state in which some $x_i=1$, the planner recommends that action, the agents get the maximum reward, and the state does not change (i.e., terminal state). In any information state in which all the realized actions are $x_i=-1$, the planner recommends an unexplored action with the highest expected reward, the agents get the maximum expected reward, and after it the state does change to include the new explored action. 

The main challenge is in the cases that the realized value of action $1$ is $x_1=0$ and $0>\mu_2>\mu_3$. In such information states we have a tension between the agent incentive, to perform action $1$ and maximize her expected reward, and the planner incentive to explore new actions to the benefit of future agents. Indeed we have two information states in which we explore stochastically, balancing between the incentives of the agent and making the recommendation BIC. In information state $\langle 0,*,*\rangle$ the planner explores with some probability action $2$, and in information state $\langle 0,-1,*\rangle$ the planner explores with some probability action $3$.

We stress that the stochastic exploration is not done in an ``independent'' way, but rather in a coordinated way through the parameter $y\in[0,1]$, which is selected initially uniformly at random, and never changes. The property that we will have is that while we are in information state $\langle 0,*,*\rangle$ we eventually have an agent that explores action $2$, and its index is $f^2(y)$. Similarly, while we are in information state $\langle 0,-1,*\rangle$ we eventually have an agent that tries action $3$, and its index is $f^3(y)$. We need to take special care to make sure that agent $f^2(y)$, which explores action $2$, is different than agent $f^3(y)$, which explores action $3$. (Clearly, each agent can explore at most one action.)
This is why we use a \textit{coordinate sampling} (to be defined later).

We also show that some information are never reachable, namely,
$\langle 0,*,1\rangle$, $\langle 0,*,-1\rangle$ and $\langle 0,1,-1\rangle$. This will be due to the fact that for any $y\in[0,1]$ we will show that $f^2(y)< f^3(y)$, which implies that we complete the exploration of action $2$ before exploring action $3$.
As we extend to $k$ actions, we use the same $y$ to coordinate between the stochastic exploration of all the actions. Then again, by showing that for any pair of actions $i<j$, it holds that $f^i(y)<f^j(y)$, we deduce that the order in which the actions are explored is from the a priori highest expected reward to the lowest, i.e., $2,3,\dots,k$.

\begin{table}[t]
\begin{tabular}{ |p{1.5cm}||p{3.9cm}|p{3.7cm}|p{1.2cm}|p{1.7cm}|}
 \hline
 \multicolumn{5}{|c|}{Recommendation Table. Policy Parameters: $(y,t)$} \\
 \hline
State& Information state &Recommendation $(\sigma_t$)& Terminal&Exploration \\
 \hline
 $S_1^{\langle*,*,*\rangle}$   & $X_1=*$    &1&&\\
 $S_t^{\langle1,*,*\rangle}$   & $X_1=1$    &1&\checkmark&\\
 $S_2^{\langle-1,*,*\rangle}$   & $X_1=-1, X_2=*$    &2&&\\
 $S_t^{\langle-1,1,*\rangle}$   & $X_1=-1, X_2=1$    &2&\checkmark&\\
 $S_3^{\langle-1,-1,*\rangle}$   & $X_1=-1, X_2=-1, X_3=*$    &3&&\\
  $S_t^{\langle-1,-1,1\rangle}$   & $X_1=-1, X_2=-1, X_3=1$    &3&\checkmark&\\
 $S_t^{\langle-1,-1,-1\rangle}$   & $X_1=-1, X_2=-1, X_3=-1$    &1&\checkmark&\\
    $S_t^{\langle0,*,*\rangle}$   & $X_1=0, X_2=*$    &$f^2_t(y)\in\{1,2\}$&&\checkmark\\
 $S_t^{\langle0,1,*\rangle}$   & $X_1=0, X_2=1$    &2&\checkmark&\\
 $S_t^{\langle0,-1,*\rangle}$   & $X_1=0,X_2=-1, X_3=*$    &$f^3_t(y)\in\{1,3\}$&&\checkmark\\
 $S_t^{\langle0,-1,1\rangle}$   & $X_1=0, X_2=-1, X_3=1$    &3&\checkmark&\\
  $S_t^{\langle0,-1,- 1\rangle}$   & $X_1=0, X_2=-1, X_3=-1$    &1
  &\checkmark&\\ 
 $S_t^{\langle0,*,-1\rangle}$   & $X_1=0, X_2=*, X_3=-1$    &infeasible[corollary \ref{infeasable}]&&\checkmark\\
 $S_t^{\langle0,*,1\rangle}$   & $X_1=0, X_2=*, X_3=1$    &infeasible[corollary \ref{infeasable}]&\checkmark&\\ 
 $S_t^{\langle0,1,- 1\rangle}$   & $X_1=0, X_2=1, X_3=- 1$    &infeasible[corollary \ref{infeasable}]&\checkmark&\\ 
 \hline
\end{tabular}\\
\caption{Algorithm 3-Action's recommendation policy\label{RecTable} } 
\end{table}

\subsection{Exploration Rates}

In this section we formalize the exploration rate that the planner can have.
A \textit{BIC exploration rate}, denoted by $q$, measures the probability that a BIC recommendation $\sigma_t=j$ is given when the planner is in some exploration state. Namely, for any BIC recommendation policy $\pi$, the BIC exploration rate is
$\sum_{S_t^{\vec{z}}\in \Gamma_t^{j-}}\Pr_{\pi}[S_t^{\vec{z}},\sigma_t=j]$,
where $\Pr_{\pi}[S_t^{\vec{z}},\sigma_t=j]$ is the probability that the planner is in $S_t^{\vec{z}}$ at time $t$, and recommends to explore action $j$, assuming that all the recommendations until the current agent use $\pi$.

Let $\hat{\pi}$ denote a BIC recommendation policy that recommends actions base on Table~\ref{RecTable} (or Table~\ref{tblExtj} for $k>3$ actions) and uses maximum BIC exploration rates for every agent $t$ and for every $j\in A$.
We explain exactly how are exploration driven recommendations assigned to agents in a way that maximizes exploration rates in subsection \ref{subsec:implementation}.\\
\textit{Maximal BIC exploration rate}, denoted by
$q_t^j$ is the maximum probability of exploration, subject to the BIC constraints, and bounded by the probability that the planner is in exploration state at time $t$ with $j$ as a recommended action. I.e., $q_t^j$ is the solution of
\begin{equation}\label{qdefpart1}
\begin{aligned}
&q_t^j=\textnormal{max}_q
& & q \\
& \textnormal{s.t.} & &
\sum_{S_t^{\vec{z}}\in \Gamma_t^{j+}}\Pr_{\hat{\pi}}[S_t^{\vec{z}}|\sigma_t=j]\;\E[u_t(j)-u_t(1)|S^{\vec{z}}_t,\sigma_t=j]+
\mu_j\frac{q}{\Pr_{\hat{\pi}}[\sigma_t=j]}\geq0\\
& & &
0\leq q\leq \sum_{S_t^{\vec{z}} \in \Gamma_t^{j-}}\Pr_{\hat{\pi}}[S_t^{\vec{z}}]_{\textstyle.}
\end{aligned}
\end{equation}
The first constraint makes sure that $\sigma_t$ is a BIC recommendation. Its first summand is a summation taken over each exploitation state probability, multiplied by the ``gain'' from choosing action $j$ instead of action $1$ in this state. The second summand is the ``loss'' of the agent, namely the prior expected reward of action $j$ (i.e., $\mu_j$), multiplied by the exploration rate $q$ and divided by the probability of the event $\sigma_t=j$ (which includes also the exploration probability $q$). The terms ``gain'' and ``loss'' are from the agent's perspective. By looking at Table ~\ref{RecTable}, we can see that when $\sigma_t=j$ is given in exploitation state, the expected utility difference is positive, therefore the agent has a ``gain'' of reward in these states. On the other hand, as we assume that $\mu_j<0$, the agent has a ``loss'' of reward in the exploration states (all of which share $x_1=0$).
When this entire expression is non-negative (i.e., the first constraint holds), it is BIC.

Notice that $\hat{\pi}$ is defined as a BIC policy, and as such every recommendation $\sigma_t=j$ is BIC, i.e., its BIC constraints must be met for every action $i\ne j$. 
We argue that in $\hat{\pi}$, if the BIC constraint of action $j$ compared to action $1$ is satisfied, all the other BIC constraints for agent $t$ are met. Therefore, we only refer to the BIC constraint with respect to action $1$ when calculating $q_t^j$.
The reason is that for any pair of actions $a<b$, and for every $y\in[0,1]$, we show that $f^a(y)<f^b(y)$, i.e., the exploration of action $a$ is done before the exploration of action $b$. 
Along with Table \ref{RecTable} that represents the recommendations of $\hat{\pi}$, we deduce that whenever a recommendation $\sigma_t=j$ is given, the reward of action $j$ is either unknown (i.e, the expected reward is $\mu_j>-1$) or $X_j$ has been observed and $x_j=1$.
Now, for any action $i$ such that $1\ne i<j$, $f^i(y)<f^j(y)$ yields that $X_i$ has been observed and $x_i=- 1$.
As for every action $j<i$, since $f^j(y)<f^i(y)$ yields that $X_i$ has not to been sampled yet, and from the assumption that $\mu_i<\mu_j$ we know that $\mu_i<\mu_j\leq \E[u_t(j)]$.
Either way $\E[u_t(\sigma_t)-u_t(i)]\geq 0$.

The second constraint in (\ref{qdefpart1}) prevents the exploration rate from exceeding the probability that the planner is in exploration state ($S_t^{\vec{z}}\in{\Gamma_t^{j-}}$). This guarantees that we can actually use of all of $q$ to give an exploration driven recommendation. Namely, $q=\sum_{S_t^{\vec{z}}\in{\Gamma_t^{j-}}}\Pr_{\hat{\pi}}[S_t^{\vec{z}},\sigma_t=j]\leq \sum_{S_t^{\vec{z}}\in{\Gamma_t^{j-}}}\Pr_{\hat{\pi}}[S_t^{\vec{z}}]$.\\
Let $n_j$ denote the index of \textit{last agent} $t\leq T$ that might explore action $j$, i.e.,
$n_j:=argmax_t(q_t^j>0)$.
For convenience, for every agent $t$ and action $j$, we denote 
\[
A_t^j:= 
\begin{cases}
0 &\quad t< j\\
\frac{2p_{j}^{1}\Pi_{i<j}p_{i}^{-1}+p_{j}^{1}\sum_{\tau=j}^{t-1}{q^j_{\tau}}}{1-2p_{j}^{1}}&\quad t\geq j\\
\end{cases}_{\textstyle ,}
\]
\[
B_t^j:=
     \begin{cases}
     0 &\quad t< j\\
     p_1^0-\sum_{m=2}^{t-1}{q_{m}^2}  &\quad t\geq j=2\\
    p_{j-1}^{- 1}\sum_{\tau=j-1}^{t-1}{q^{j-1}_{\tau}}-\sum_{\tau=j}^{t-1}{q^j_{\tau}}\quad  &\quad t\geq j\geq 3\\
     \end{cases}_{\textstyle .}
\]
\subsection{Computing the Maximum BIC Exploration Rates}
\label{sec:maxBIC}

We now calculate the maximum BIC exploration rates. 
(The next lemma's proof for $k=3$, namely, $q_t^2$ and $q_t^3$, is in Appendix \ref{app:sec3}, and the proof for $k>3$, i.e., any $q_t^j$, is in Appendix \ref{appendix1}.)
\begin{lemma}\label{BigThm}
Given $q_2^2, \ldots,q_{t-1}^2$, we have
\begin{equation}
q_{t}^2 = 
     \begin{cases}
       0 &\quad t=1\\
       \min (A^2_t,B^2_t) &\quad t\geq 2\\
     \end{cases}_{\textstyle ,}
\end{equation}
And for action $j\geq 3$, given $q^i_\tau$ for $i\leq j-1$ and $\tau\leq t-1$, assuming $q_{t-1}^{j-1}=A_{t-1}^{j-1}$ and $t\leq n_{j-1}$, we have
\begin{equation}
q^j_{t} = 
     \begin{cases}
       0 &\quad t<j\\
       \min(A^j_t,B^j_t) &\quad t\geq j\\
     \end{cases}_{\textstyle .}
\end{equation}
In addition we show that $q_{t}^{j}\leq  p_{j-1}^{-1}A_{t-1}^{j-1}$.
\end{lemma}

The next lemma derives the value of $q^j_t$ (without an assumption on $q^{j-1}_{t-1}$).
\begin{lemma}
\label{qAfternj}
For action $j\geq 3$, given $q^i_\tau$ for $i\leq j-1$ and $\tau\leq t-1$, such that $t>n_{j-1}$, we have
\begin{equation}
q^j_{t} = 
     \begin{cases}
       0 &\quad t<j\\
       \min(A^j_t,B^j_t) &\quad t\geq j\\
     \end{cases}_{\textstyle .}
\end{equation}
\end{lemma}
The following are consequences of Lemma \ref{BigThm} and Lemma \ref{qAfternj}.
\begin{lemma}\label{corq3bigger0}\label{qjPositive}
For every $j>1$, the exploration rate of agent $j$ for action $j$ is strictly positive, i.e., $q_{j}^j>0$.
\end{lemma}
The following lemmas relate the exploration rates $q^j_t$ and the parameters $A^j_t$ and $B^j_t$. 
\begin{lemma}\label{APositive}
For every action $j$ and agent $t> j$ and that $q^j_{t-1}>0$, it holds that $A_j^j\leq A_{t-1}^j< A_{t}^j$.
\end{lemma}\label{incAtj}

\begin{lemma}
\label{BigThm2}
For every action $j$, for every $n_j\geq t>n_{j-1}$, it holds that 
\begin{enumerate}\label{props}
\item $B_{t-1}^j > B_t^j\geq 0$.
\item $q_t^j>0$.
\item If $q_t^j=B_t^j(>0)$, then it holds that $q_{t+i}^j=B_{t+i}^j=0$ for every $i\geq1$, therefore $t=n_j$ and we stop.
\end{enumerate}
\end{lemma}

\subsection{Properties of the Exploration Rates}
The policy assigns exploration driven recommendations of actions to agents while maximizing all exploration rates simultaneously.
In this subsection we show properties regarding the exploration rates which later enable to have correlated randomization between agents and actions. It would also help us to show in subsection \ref{subsec:implementation} and section \ref{OptimalitySec} that as a BIC policy that recommends actions base on Table~\ref{RecTable} (or Table~\ref{tblExtj} for $k>3$ actions) and maximizes exploration rates, $\hat{\pi}$: (1) has a well-defined implementation (that assigns a single action for every agent), (2) eventually reaches a terminal state, and (3) maximizes expected social welfare.
The following theorem is a corollary to Lemmas \ref{BigThm} - \ref{BigThm2}, and states the exact exploration rates.
\begin{theorem}\label{qBeforeLast}
For action $2$ and for agent $t$ we have,
\[ 
q_{t}^2 = 
     \begin{cases}
       0 &\quad t=1\\
       \frac{2p_{1}^{-1}p_{2}^{1}+p_{2}^{1}\sum_{m=2}^{t-1}{q^2_{m}}}{1-2p_{2}^{1}} &\quad 2\leq t<n_2\\
        p_1^0-\sum_{m=2}^{n_2-1}{q_{m}^2}&\quad t=n_2\\
0 &\quad t>n_2\\
     \end{cases}_{\textstyle ,}
\]
for action $j\geq 3$ and agent $t$ we have,
\[   
q_t^j = 
     \begin{cases}
       0  &\quad  t<j\\
  \frac{2p_{j}^{1}\Pi_{i<j}p_{i}^{-1}+p_{j}^{1}\sum_{\tau=j}^{t-1}{q^j_{\tau}}}{1-2p_{j}^{1}}  &\quad j\leq t<n_j\\
       p_{j-1}^{-1}\sum_{\tau=j-1}^{n_{j}-1}{q^{j-1}_{\tau}}-\sum_{\tau=j}^{n_j-1}{q^j_{\tau}}\quad  &\quad t=n_j\\
       0  &\quad t>n_j
     \end{cases}_{\textstyle .}
\]
\end{theorem}
Let $\rho_j=p_1^0\Pi_{i=2}^{j-1}p_i^{- 1}$. 
We show that $\rho_j$ is the total exploration rate of action $j$.
\begin{lemma}\label{lemmaRho}
For $T\geq n_k$, the probability for exploration driven recommendation for any action $j$ is $\rho_j$, i.e.,
\[\Pr[\exists t: \sigma_t=j, S_t^{\vec{z}}\in \Gamma^{j-}_t]=\rho_j.\]
\end{lemma}
\subsection{Determining the explorers- correlating across agents and actions}\label{subsec:implementation}
We now explain how the algorithm chooses which agent will explore each action.
Recall that the planner knows the history $h_{t-1}$, and therefore knows the current state at time $t$, as defined in Table ~\ref{RecTable}. Table ~\ref{RecTable} has clear recommendation for any exploitation state (i.e., any state that is not marked as exploration). 

Before explaining how the policy chooses recommendation for the exploration states $q_t^j$ with a special correlated randomization technique, we want to point out two problems that occur by simply recommending agent $t$ to explore action $j$ with probability $q_t^j$ independently of the other agents and actions. First, is that agents might be sampled to explore two (or more) actions, which will delay the exploration of all but one of the actions, as the planner can only recommend one action per agent. Second, is that action $j$ might be explored before the planner knows the rewards of $2,\dots,j-1$, which we later show that is not optimal.

We now return to describe how the planner should select which action to recommend based on Table \ref{RecTable}. Knowing the current information state, the planner sets $\pi_t(h_{t-1}) = v_t$ to be the corresponding recommendation for this state in Table ~\ref{RecTable}. Together with the policy parameters and the functions $f^j$ that we later define in Definitions \ref{defInput} and \ref{defF}, respectively, she returns $\sigma_t \sim v_t$ as the recommended action.

\begin{definition}\label{defInput} a valid input for our algorithm is a triple $\langle y, t, Q\rangle$ such that:
\begin{enumerate}
\item $y\in[0,1]$ is a real number that is sampled from a uniform distribution in $[0,1]$.
\item $t$ indicates the agent number (the agent for which the algorithm is run).
\item $Q=\{q^j|j\in\{2,...,k\}\}$ is a set that contains exploration rates vectors for each action excluding action $1$, such that $q^j[t]:=q_t^j$ (i.e, the exploration rate for agent $t$ with $j$ as recommended action). 
\end{enumerate}
\end{definition}
We now define the functions $f^j$ that determines which agent will explore action $j$.
\begin{definition}\label{defF}
Let $f^j:[0,1]\rightarrow \{1,...,n_j\}$ be the function that maps a real number $y\in[0,1]$ to an agent $t$ such that 
\[
f^j(y)=\textnormal{argmax}_{t}(\sum_{\tau=1}^{t-1}q_{\tau}^j<y\rho_j)_{\textstyle .}
\]
Let $f^j_t:[0,1]\rightarrow \{1,j\}$ be the function for action $j$ and agent $t$ that maps a real number $y\in[0,1]$ to a recommendation for agent $t$, i.e., $\sigma_t$, and is defined as follows:
\[   
f_t^j(y) := 
     \begin{cases}
       j &\quad f^j(y)=t\\
       1 &\quad else\\
     \end{cases}_{\textstyle .}
\]
\end{definition}

The following lemma shows that different actions are explored by different agents, and that better a priori actions are explored always earlier.

\begin{lemma}\label{f2smallest}
For every $y\in [0,1]$, and for every action $j$, $f^{j}(y)<f^{j+1}(y)$.
\end{lemma}
Since $f^j(y)=n_j$, Lemma \ref{f2smallest} implies the following corollaries.
\begin{corollary}\label{lastAgnetInd}
For every action $j$, it holds that $n_j+1\leq n_{j+1}$.
\end{corollary}
\begin{corollary}\label{infeasable} Action $j$ is explored before any action $i>j$, making every state $S_t^{\vec{z}}$ such that $S_t^{\vec{z}}[j]=*$ and $S_t^{\vec{z}}[i]\ne *$ (e.g., $S_t^{\langle0,*,-1\rangle}$) infeasible for every agent $t$. Namely,
\[\Gamma^{j-}_t=\{S_t^{\vec{z}}|S_t^{\vec{z}}[1]=0,\forall 1<i<j: S_t^{\vec{z}}[i]=- 1, \forall i\geq j:  S_t^{\vec{z}}[i]=*]\}_{\textstyle .}\]
Hence, $|\Gamma^{j-}_t|=1$.
\end{corollary}
We finish this section by showing that the policy $\hat{\pi}$ has a well-defined implementation in a sense that every agent gets exactly one action as a recommendation in Theorem \ref{welldefined}.
\begin{theorem}\label{welldefined}
Recommendation policy 
$\hat{\pi}$ has a well-defined implementation recommendation policy, since for every $y\in[0,1]$, and for pair of actions action $i\ne j$, $f^i(y) \ne f^{j}(y)$, and there exists $t\in \{1,\ldots,n_j\}$ such that $f^j(y)=t$.
This implies that every agent receives a recommendation for exactly one action.
\end{theorem}

\section{Optimality}\label{OptimalitySec}

\subsection{Finite exploration}
Clearly the flow between the information states is acyclic. 
From Table \ref{RecTable}, when the planner is in a non-terminal state, she is exploring, with some probability.
This implies that after $n_k$ (the last agent that might explore action $k$), there is no more exploration. Therefore,
$\hat{\pi}$ will eventually reach a terminal state and thus complete the exploration. From this we derive the following theorem:
\begin{theorem}\label{lemFinishExp}
By using the policy $\hat{\pi}$, as long as the planner has not observed an action $j$ with $x_j=1$, she will keep exploring until all actions' rewards are revealed. Therefore $\hat{\pi}$ always reaches a terminal state.
Formally, if  at time $t$ action $j$ is the a priory best action with an unknown reward,  ${\vec{z}}[j]=*$, and for every action $j'$ does not have optimal reward, $\vec{z}[j']\ne 1$ implies that $\Pr_{\hat{\pi}}[\sigma_t=j]>0$ and therefore $S_{n_k+1}^{\vec{z}}$ is terminal. 
\end{theorem}
\subsection{Minimum exploration time}

Two BIC planners may differ only in their recommendations when they are in the exploration states. We would prefer the one that explores the actions "faster", as it would mean finding the optimal action sooner. For this we define a partial order between policies. We say that a policy is \textit{stochastic dominant} over another if it discovers the realizations of the rewards faster.
\begin{definition} A BIC policy algorithm $\pi_A$ is \textit{stochastic dominant} over another BIC policy algorithm $\pi_B$ if for every prior $D$ and for every agent $t$, $\pi_A$ has at least the same probability to observe action $j$'s reward as $\pi_B$, and for some action $j$ a strictly higher probability to know it's reward in time $t$. I.e., for any agent $t$ and action $j$ we have
$\Pr_{\pi_A}[S_t^{\vec{z}}\wedge \vec{z}[j]\ne *]\geq \Pr_{\pi_B}[S_t^{\vec{z}}\wedge \vec{z}[j]\ne *]$, and there exists some action $j$ and agent $t$ for which $\Pr_{\pi_A}[S_t^{\vec{z}}\wedge \vec{z}[j]\ne *]>\Pr_{\pi_B}[S_t^{\vec{z}}\wedge \vec{z}[j]\ne *]$.
\end{definition} 
The following lemma
states that the suggested policy, $\hat{\pi}$ is stochastic dominant over all other BIC policies.
\begin{lemma}\label{stoDom}
Let $\pi_A\ne \hat{\pi}$ be a BIC policy algorithm, with the same recommendations for the exploitation states as in Table ~\ref{RecTable}. Then $\hat{\pi}$ is stochastic dominant over $\pi_A$.
\end{lemma}

From Lemma \ref{stoDom} we easily obtain that $\hat{\pi}$ maximizes exploration rates of each action $j$ and agent $t$. Due to the use of $y\in[0,1]$ to decide which agent will explore each action, $\hat{\pi}$ manages to maximize exploration rates of all the actions independently. This give us an important result regarding $\hat{\pi}$:
\begin{theorem}\label{corMinTerminal}
The policy $\hat{\pi}$ minimizes the time until terminal state .
\end{theorem}

\subsection{Maximum expected social welfare}
In this section we present the main result: The best BIC policy is the one that minimizes exploration time for every action simultaneously.
\begin{theorem}\label{zeroReg}
Let $\pi_{opt}$ be a BIC policy algorithm that maximizes the expected Social welfare. Then for a large number of agents (specifically, $T\geq \lceil{\frac{1}{p_{k}^1}+n_k-1}\rceil$), it holds that
\[
SW_{T}(\pi_{opt})=SW_{T}(\hat{\pi})_{\textstyle .}
\]
\end{theorem}
From the above theorem we deduce the following corollary.
\begin{corollary}\label{maxSW}
Recommendation policy
$\hat{\pi}$ maximizes social welfare for every $T\geq \lceil{\frac{1}{p_{k}^1}+n_k-1}\rceil$.
\end{corollary}
\subsection{Limited number of agents}
The planner's goal is to maximize social welfare. If there is a limited number of agents, she cannot rely on the existence of the agent that balances the loss of social welfare (i.e., agent $t_1$ in the proof for Theorem \ref{zeroReg}).
Our algorithm must be adjusted for that. A natural solution is to limit the recommendation for exploration, so that the planner must give exploration driven recommendation for action $j$ in round $t$ if the gain for the following agents, $p_j^1(T-t-1)$ is high enough to cover for the expected loss of the $t$-th agent, i.e., $\mu_j$. We add the following requirement that must be fulfilled if the algorithm gives an exploration driven recommendations to agent $t$. Namely,
\[
(T-t)\cdot p_{j}^{1}+\mu_j\geq 0 _{\textstyle .}
\]
Alternatively, $(T-t+2)\cdot p_{j}^{1}\ge 1$.

Theorem \ref{zeroReg}'s proof still applies for any pair $\langle j,t\rangle$ that meets the additional requirement. For pairs $\langle j,t\rangle$ that do not meet the requirement, action $j$ is no longer recommended for exploration in round $t$ or afterwards. An exploration driven recommendation for these agents harms the social welfare.
\section{Continuous distribution for the a priori best action's reward}\label{contSec}
In this section we explore the same model with one significant difference. The prior distribution $D_1$ is now a continuous distribution that has full support of $[-1,1]$. (Note that we do not allow mass points.) Consider that the number of agents, $T$, is large enough so that a social planner must complete the exploration of all the actions.

The different type of recommendation policy algorithm we introduce for this setting is a generalization of the partition policy, originally defined in Kremer et al. \cite{Kremer-JPE14}.

\subsection{Partition policy as a recommendation algorithm}
The following two definitions are used to define a partition policy (in Definition \ref{partitionPolicyDef}).
\begin{definition}\label{collection}
$\Theta^j$ is a collection of $T$ \textbf{disjoint} sets, 
$\Theta^j:=\{\theta_t^j\}_{t=2}^{T}$, where $\theta_t^j\subseteq[-1,1]$.
\end{definition}
\begin{definition}\label{noIntersection}
A valid input for any partition policy algorithm is a series $\Theta=(\Theta^j)_{j=2}^k$ s.t. for any pair of actions $i\ne j\in A$, it holds that $\theta_t^j \cap\theta_t^i=\emptyset$ for every agent $t$.
\end{definition}
\begin{definition}\label{partitionPolicyDef}
Given a valid input, $(\Theta^j)_{j=2}^k$, and a realization $X_1=x_1$, a \textit{partition policy} is a recommendation policy that makes the following recommendations. For agent $t$ we have,

\begin{enumerate}
\item\label{sigma11} For $t=1$ we have $\sigma_1=1$.
\item\label{optimalRewFound} If there is an explored action $j$ with a reward of $1$ (i.e., it is optimal), then $\sigma_t=j$.
\item\label{diffBullet} Else, if $x_1\in \theta_t^j$ then $\sigma_t=j$.
(In this case agent $t$ is the first agent for whom $\sigma_t=j$.).
\item\label{mustInclude} $[-1,\mu_j]\subseteq\theta_j^j$.
\item\label{defultExploit} Else, $\sigma_{t}=1$.
\end{enumerate}
\end{definition}
Let us inspect each clause in the above definition with regards to BIC and social welfare.
\begin{enumerate}
\item Since action $1$ is the a priori better action, any BIC policy must recommend to agent $t=1$ action $\sigma_1=1$ (clause (\ref{sigma11})).
\item After finding an explored action with value $1$, to maximize the social welfare we must recommend it (clause (\ref{optimalRewFound})).
\item Clause (\ref{diffBullet}) is where action $j$ is recommended for the first time. 
Once the planner will observe the value of the explored action, if it is $1$ it will be recommended to all future agents (as indicated in clause \ref{optimalRewFound}).
\item Clause (\ref{mustInclude}) deals with the case that the a priori best action has low value. Together with clause \ref{optimalRewFound}, it guarantees that agent $j$ performs action $j$ if none of the explored actions has a higher reward.
\item Clause (\ref{defultExploit}) gives an exploitation recommendation. Note that any explored action $j$, in this case, has $x_1\geq x_j=-1$.
\end{enumerate}
Notice that a valid input, $(\Theta^j)_{j=2}^k$, insures that every agent $t\geq 2$ receives a recommendation for exactly one action.
We can now derive the following lemma:
\begin{lemma}\label{BICisPartition}
The optimal BIC recommendation policy is a partition policy.
\end{lemma}
\subsection{The suggested BIC Partition Policy}

Recall that agent $t$ finds that recommendation $\sigma_t$ to be BIC if for any action $i\in A$ we have
\[
\E[u_t(j)-u_t(i)|\sigma_t=j]\geq 0.
\]
Note that this holds if and only if for any action $i\in A$ 
\[
\Pr[\sigma_t=j]\cdot \E[u_t(j)-u_t(i)|\sigma_t=j]\geq 0.
\]
Namely, 
\[
\int_{\sigma_t=j} [X_j-X_i]dD\geq 0.
\]

We now describe how to extract parameters for the suggested policy iteratively, given a prior $D$ for the problem. Then we continue by showing that these collections can be used as a valid input of a partition policy. Finally, we show that using these parameters produces a BIC recommendation policy.

\begin{definition}\label{defIntervals}
The sets $\hat{\Theta}^j=\{\hat{\theta}_t^j\}_{t=2}^{T+1}$ are calculated as follows.\\
Let $\hat{\theta}_t^j$ be the ordered interval $(i_{t}^{j},i_{t+1}^j]$, where: 
\begin{itemize}
\item For $t<j$, 
$\hat{\theta}_t^j=\emptyset$ (this can be done by setting 
$i_t^j:=-1$ for $t\leq j$).
\item For $t=j$, recall that $i_j^j=-1$, and let $\omega_{j+1}^j$ be the solution to:
\begin{equation}\label{defijj}
\prod_{n=2}^{j-1}(p_n^{-1})\int_{\mu_j\geq X_1} [\mu_j-X_1]dD_1=
\int_{\mu_{j}\leq X_1\leq \omega_{j+1}^j} [X_1-\mu_j]dD_1.
\end{equation}
\item For every $t>j$, let $\omega_{t+1}^j$ be the solution to:
\begin{equation}\label{defitj}
p_j^1\int_{-1\leq X_1\leq i_{t}^j} [1-X_1]dD_1=
\int_{i_{t}^j< X_1\leq \omega_{t+1}^j} [X_1-\mu_j]dD_1.
\end{equation}
\item $i_{t+1}^j=min(1,\omega_{t+1}^j)$ for every $t\geq j$. (We have $\hat{\theta}_t^j =(i_{t}^{j},i_{t+1}^j]$.).
\end{itemize}
\end{definition}

Notice that in each step, distribution $D$, the parameters $p_j^1, p_n^{-1}, \mu_j$ and $i_t^j$ are known, therefore one can compute the value of $i_{t+1}^j$.\\
In the next lemma, we show that $i_{t}^{j}\leq i_{t+1}^j$ for every action $j$ and agent $t$. This will allow us to deduce in Corollary \ref{wellDefinedPart} that $\hat{\Theta}_j$ is a collection of disjoint sets, which is required from a valid input for partition policy. 
\begin{lemma}\label{monoincI}
For every action $j\ne 1$ and agent $t\in T$, it holds that $i_{t}^{j}\leq i_{t+1}^j$.
\end{lemma}
\begin{corollary}\label{wellDefinedPart} For every action $j$, since there is no intersection between the ordered intervals $(i_{t}^{j},i_{t+1}^j]$ for every agent $t$, $\hat{\Theta}_j$ is a collection of disjoint sets.
\end{corollary}
For $(\hat{\Theta^j})_{j=2}^k$ to be well defined we need to verify that there is no intersection for the same agent $t$ for different actions (as stated in Definition \ref{noIntersection}). In the next lemma we show that for every agent and action $\langle t,j\rangle$ the right bound of $\theta_{t}^{j+1}$ is smaller than the left bound of $\theta_{t}^{j}$.
\begin{lemma}\label{incI}
For every $t\geq j$, it holds that $i_{t}^j\geq i_{t+1}^{j+1}$, and there is an equality only if $i_{t}^j= i_{t+1}^{j+1}=1$.
\end {lemma}
Let $\hat{\pi}$ be a partition policy that uses the suggested parameters, $(\hat{\Theta^j})_{j=2}^k$, as an input.
\begin{corollary}
$\hat{\pi}$ is well defined, as for every $x_1\in[-1,1]$ there exists exactly one action $j\in A$ such that $\sigma_t=j$.
\end{corollary}
We now investigate what is required from a BIC partition policy, in order to show that the suggested partition policy is BIC. Following the same exhaustion demonstrated in Example \ref{exampleSigmat2}, for every agent $t\geq2$ a BIC recommendation for action $j\ne 1$ can be either exploration driven or exploitation driven:
\begin{enumerate}
\item \textbf{Exploitation driven recommendation}: Action $j\ne 1$ is the best action given the history. Once again, one of the following holds:
\begin{enumerate}
\item \textbf{A known reward}- Action $j$ has the best possible realized value (i.e., $x_{j}=+1$) after one of the previous agents $j\leq \tau<t$ has explored it.
Formally, the expected "gain" of agent $t$ from choosing action $j$
over of action $1$ in this case is
\[
\int_{X_j=1,X_1\in{\cup_{\tau<t}\theta_j^{\tau}}} [X_j-X_1]dD=
p_j^1\int_{X_1\in{\cup_{\tau<t}\theta_j^{\tau}}} [1-X_1]dD_1.
\]
\item \textbf{An unknown reward}- 
The observed realization $x_1$ yields lower
reward compared to the prior expected reward of action $j$, i.e., $x_1<\mu_j(<\dots<\mu_2)$, and the better a priori actions, $k<j$ have been explored and resulted in minimal reward of $-1$.
Thus action $j$ is better to execute then action $1$, and the expected "gain" here is
\[
\prod_{n=2}^{j-1}(p_n^{-1})\int_{\mu_j\geq X_1} [\mu_j-X_1]dD_1.
\]
\end{enumerate}

\item \textbf{Exploration driven recommendation}: The planner has not yet observed an action with the best possible reward (i.e., $+1$), and $x_1>\mu_j$. This recommendation does not benefit with agent $t$ (but might yield higher expected social welfare hence is possible by the planner). The expected "loss" in this case is:
\[
\int_{\mu_j<X_1,X_1\in{\theta_j^t}} [X_j-X_1]dD=
\int_{\mu_j<X_1,X_1\in{\theta_j^t}} [\mu_j-X_1]dD_1.
\]
\end{enumerate}
Taken together, the summation of the expected gains and loss is equivalent to the left hand sides of the integrals in the next lemma.
\begin{lemma} A recommendation $\sigma_t=j\ne 1$ is BIC w.r.t. action $1$ if for $t=j$ we have,
\begin{equation}\label{twoPartInt1}
\int_{\mu_j<X_1,X_1\in{\theta^j_t}} [\mu_j-X_1]dD+
\prod_{n=2}^{j-1}(p_n^{-1})\int_{\mu_j\geq X_1} [\mu_j-X_1]dD_1
\geq 0_{\textstyle ,}
\end{equation}
and for $t>j$ we have,
\begin{equation}\label{twoPartInt}
\int_{X_1\in{\theta^j_t}} [\mu_j-X_1]dD+
p_j^1\int_{X_1\in{\cup_{\tau<t}\theta_j^{\tau}}} [1-X_1]dD_1\geq 0_{\textstyle .}
\end{equation}
Notice that for $t>j$, $X_1\in{\theta^j_t}$ yields $X_1\geq \mu_j$ since the sets are disjoint and according to clause (\ref{mustInclude})
in definition \ref{partitionPolicyDef}, $[-1,\mu_j]\subseteq\theta_j^j$.
\end{lemma}

Let us compare the BIC constraints in above lemma to the suggested series, $\hat\Theta$.\\
If $i_{t}^{j}=1$, then from Lemma \ref{monoincI} and the definition of $i_{t}^{j}$, it holds that $i_{t+1}^{j}=1$, which yields $\hat{\theta}_t^j=\emptyset$.\\
Since both parts of (\ref{defitj}) are zero in this case, so the (exploitation driven) recommendation $\sigma_t=j$ is BIC.
Else, $i_{t}^{j}=\omega_t^j$, in which case (\ref{defijj}) and (\ref{defitj}) satisfy (\ref{twoPartInt1}) and (\ref{twoPartInt}) respectively. This allows us to derive the following corollary.
\begin{corollary}\label{BICwrt1}
The partition policy $\hat{\pi}$ that uses the series $\hat\Theta$ as input, outputs BIC recommendations $\sigma_t=j$ w.r.t. action $1$ for every agent $t$.
\end{corollary}

Due to exploration of the a priori better actions earlier, it is sufficient for the suggested partition policy, $\hat{\pi}$ to maintain the BIC constraint w.r.t. action $1$ alone. I.e., if $\sigma_t=j$ is an exploration driven recommendation, then $j$ has the lowest index of an unexplored action (and therefore has the highest expected value among all unknown rewards.). Agent $t$ would not prefer any other action $i\ne j$ as such a behavior would yield either a
lower prior reward for actions $i>j$ or the lowest possible reward (i.e., $-1$) for actions $i<j$.
The following are corollaries of the last lemma.

\begin{theorem}
$\hat{\pi}$ is a BIC partition policy.
\end{theorem}
\begin{corollary}\label{ActionsbyOrder}
$\hat{\pi}$ recommends the actions in ascending order, i.e., for every $j<i$, action $j$ is explored before action $i $.
\end{corollary} 
\subsection{Optimality}
We follow the same logic process of the proofs provided in Section \ref{OptimalitySec} with some tiny adjustments. The abstraction of the information states is used here with one small difference- considering all the BIC policies recommend action $1$ to the first agent and observe reward $x_1\in[-1,1]$, we update the definition of \textit{stochastic dominant} to be conditioned on $x_1$ as follows.
\begin{definition} A BIC policy algorithm $\pi_A$ is \textit{stochastic dominant} over another BIC policy algorithm $\pi_B$ if for every prior $D$ and a realization $x_1\in[-1,1]$, and for any agent $t$, $\pi_A$ has at least the same probability to observe action $j$'s reward as $\pi_B$, and for some action $j$ a strictly higher probability to observe it's reward in time $t$. I.e., for any agent $t$, action $j$ and realization $x_1\in[-1,1]$ we have
$\Pr_{\pi_A}[S_t^{\vec{z}}\wedge \vec{z}[j]\ne *|\vec{z}[1]=x_1]\geq \Pr_{\pi_B}[S_t^{\vec{z}}\wedge \vec{z}[j]\ne *|\vec{z}[1]=x_1]$, and there exists some agent $t$ and action $j$ for which $\Pr_{\pi_A}[S_t^{\vec{z}}\wedge \vec{z}[j]\ne *|\vec{z}[1]=x_1]>\Pr_{\pi_B}[S_t^{\vec{z}}\wedge \vec{z}[j]\ne *|\vec{z}[1]=x_1]$.
\end{definition} 
The next step is to show the partition policy with the suggested parameters as input yields stochastic dominance. It is done in the following lemma.
\begin{lemma}\label{stoDomcont}
Let $(\Theta_j)_{j=2}^k$, and the realization $X_1=x_1$, be the input for a (BIC) \textit{partition policy} $\pi_A\ne \hat{\pi}$. Then $\hat{\pi}$ is stochastic dominant over $\pi_A$.
\end{lemma}
From Lemma ~\ref{stoDomcont} we deduce that $\hat{\pi}$ maximizes exploration for each action $j$ and agent $t$. Due to the use of disjoint sets in the definition of a partition policy, $\hat{\pi}$ manages to maximize exploration rates independently. This gives us the important result of time minimization:
\begin{theorem}
$\hat{\pi}$ minimizes the time until terminal state.
\end{theorem}
Finally, we show the main result for this case as well.
\begin{theorem}\label{swCont}
$\hat{\pi}$ maximizes social welfare for unlimited number of agents.
\end{theorem}

\section{Conclusion and open problems}

This paper explores the problem of incentivizing exploration via Bayesian persuasion. 
We consider two different supports for the a priori better action, a discrete version $\{-1,0,1\}$, and a continuous version $[-1,1]$.
In both settings, our optimal policy explores the better a priori actions earlier.
In addition, it maximizes the exploration, subject to the BIC constraints.
This leads to a planner policy that maximize the social welfare.

Our optimal policy also achieves both: (1) minimizing the time until all of the actions are explored, and (2) that all the actions are explored, in case of large enough $T$. 
Our optimal policy requires special correlated randomization to guarantee the optimality.

There are few obvious open problems, First, to extend the support of {\em all} actions to be, for example $[-1,+1]$. Second, to consider stochastic actions, even simple Bernoulli random variables with different probabilities. Third, enabling the agents to receive limited amount of information about the past. (The challenge here is to make the information informative, and still allow the planner to explore all actions, eventually.)

\subsection{The case of continuous distribution for all actions' reward}
Let us discuss a case in which the support of each action's reward is $[-1,1]$. In order to tackle such setting, a change in Definition \ref{partitionPolicyDef} is required.
First, unless all rewards are known there is no action with optimal reward, hence  clause (\ref{optimalRewFound}) must be removed.
Second, clause (\ref{defultExploit}) should be changed to  $\sigma_t=\textnormal{argmax}_j{x_j}$, so that exploitation driven recommendation with known reward would have the best reward.\\
The issue is that while the BIC constraints w.r.t. the a priori best action hold in such policy, the BIC constraints w.r.t. the other actions do not (automatically) apply anymore. 
Consequently, one must also update how the algorithm selects which agent should explore which action.
For example, if agent $5$ receives a recommendation for action $3$ (i.e., $\sigma_5=3$) and there is a very low probability that either $x_2$ or $x_1$ are smaller than $\mu_3$ and $\mu_3$ is slightly bigger than $-1$. In this case, agent $5$ knows that it is highly unlikely for such a recommendation to benefit her, and might prefer action $2$ over this recommendation.

\bibliographystyle{ACM-Reference-Format}
\bibliography{refs}


\begin{thebibliography}{14}


\ifx \showCODEN    \undefined \def \showCODEN     #1{\unskip}     \fi
\ifx \showDOI      \undefined \def \showDOI       #1{#1}\fi
\ifx \showISBNx    \undefined \def \showISBNx     #1{\unskip}     \fi
\ifx \showISBNxiii \undefined \def \showISBNxiii  #1{\unskip}     \fi
\ifx \showISSN     \undefined \def \showISSN      #1{\unskip}     \fi
\ifx \showLCCN     \undefined \def \showLCCN      #1{\unskip}     \fi
\ifx \shownote     \undefined \def \shownote      #1{#1}          \fi
\ifx \showarticletitle \undefined \def \showarticletitle #1{#1}   \fi
\ifx \showURL      \undefined \def \showURL       {\relax}        \fi
\providecommand\bibfield[2]{#2}
\providecommand\bibinfo[2]{#2}
\providecommand\natexlab[1]{#1}
\providecommand\showeprint[2][]{arXiv:#2}

\bibitem[\protect\citeauthoryear{Bahar, Smorodinsky, and Tennenholtz}{Bahar
  et~al\mbox{.}}{2016}]%
        {BaharST16}
\bibfield{author}{\bibinfo{person}{Gal Bahar}, \bibinfo{person}{Rann
  Smorodinsky}, {and} \bibinfo{person}{Moshe Tennenholtz}.}
  \bibinfo{year}{2016}\natexlab{}.
\newblock \showarticletitle{Economic Recommendation Systems: One Page
  Abstract}. In \bibinfo{booktitle}{\emph{Proceedings of the 2016 {ACM}
  Conference on Economics and Computation, {EC}}}. \bibinfo{publisher}{ACM},
  \bibinfo{address}{New York, NY, USA}, \bibinfo{pages}{757}.
\newblock
\urldef\tempurl%
\url{https://doi.org/10.1145/2940716.2940719}
\showDOI{\tempurl}


\bibitem[\protect\citeauthoryear{Cesa-Bianchi and Lugosi}{Cesa-Bianchi and
  Lugosi}{2006}]%
        {regret-book}
\bibfield{author}{\bibinfo{person}{Nicol{\`o} Cesa-Bianchi} {and}
  \bibinfo{person}{Gabor Lugosi}.} \bibinfo{year}{2006}\natexlab{}.
\newblock \bibinfo{booktitle}{\emph{Prediction, learning, and games}}.
\newblock \bibinfo{publisher}{Cambridge University Press},
  \bibinfo{address}{New York, NY, USA}.
\newblock
\urldef\tempurl%
\url{https://doi.org/10.1017/CBO9780511546921}
\showDOI{\tempurl}


\bibitem[\protect\citeauthoryear{Che and H{\"o}rner}{Che and
  H{\"o}rner}{2013}]%
        {Che-13}
\bibfield{author}{\bibinfo{person}{Yeon-Koo Che} {and}
  \bibinfo{person}{Johannes H{\"o}rner}.} \bibinfo{year}{2013}\natexlab{}.
\newblock \showarticletitle{Optimal Design for Social Learning}.
\newblock \bibinfo{journal}{\emph{SSRN Electronic Journal}}
  (\bibinfo{year}{2013}).
\newblock


\bibitem[\protect\citeauthoryear{Dughmi, Kempe, and Qiang}{Dughmi
  et~al\mbox{.}}{2016}]%
        {DughmiKQ16}
\bibfield{author}{\bibinfo{person}{Shaddin Dughmi}, \bibinfo{person}{David
  Kempe}, {and} \bibinfo{person}{Ruixin Qiang}.}
  \bibinfo{year}{2016}\natexlab{}.
\newblock \showarticletitle{Persuasion with Limited Communication}. In
  \bibinfo{booktitle}{\emph{Proceedings of the {ACM} Conference on Economics
  and Computation, {EC}, Maastricht, The Netherlands, July 24-28}}.
  \bibinfo{publisher}{ACM}, \bibinfo{address}{New York, NY, USA},
  \bibinfo{pages}{663--680}.
\newblock
\urldef\tempurl%
\url{https://doi.org/10.1145/2940716.2940781}
\showDOI{\tempurl}


\bibitem[\protect\citeauthoryear{Dughmi and Xu}{Dughmi and Xu}{2016}]%
        {DughmiX16}
\bibfield{author}{\bibinfo{person}{Shaddin Dughmi} {and}
  \bibinfo{person}{Haifeng Xu}.} \bibinfo{year}{2016}\natexlab{}.
\newblock \showarticletitle{Algorithmic Bayesian persuasion}. In
  \bibinfo{booktitle}{\emph{Proceedings of the 48th Annual {ACM} {SIGACT}
  Symposium on Theory of Computing, {STOC}}}. \bibinfo{publisher}{ACM},
  \bibinfo{address}{New York, NY, USA}, \bibinfo{pages}{412--425}.
\newblock
\urldef\tempurl%
\url{https://doi.org/10.1145/2897518.2897583}
\showDOI{\tempurl}


\bibitem[\protect\citeauthoryear{Dughmi and Xu}{Dughmi and Xu}{2017}]%
        {DughmiX17}
\bibfield{author}{\bibinfo{person}{Shaddin Dughmi} {and}
  \bibinfo{person}{Haifeng Xu}.} \bibinfo{year}{2017}\natexlab{}.
\newblock \showarticletitle{Algorithmic Persuasion with No Externalities}. In
  \bibinfo{booktitle}{\emph{Proceedings of the 2017 {ACM} Conference on
  Economics and Computation, {EC}, Cambridge, MA, USA, June 26-30, 2017}}.
  \bibinfo{publisher}{ACM}, \bibinfo{address}{New York, NY, USA},
  \bibinfo{pages}{351--368}.
\newblock
\urldef\tempurl%
\url{https://doi.org/10.1145/3033274.3085152}
\showDOI{\tempurl}


\bibitem[\protect\citeauthoryear{Frazier, Kempe, Kleinberg, and
  Kleinberg}{Frazier et~al\mbox{.}}{2014}]%
        {FKKK14}
\bibfield{author}{\bibinfo{person}{Peter Frazier}, \bibinfo{person}{David
  Kempe}, \bibinfo{person}{Jon~M. Kleinberg}, {and} \bibinfo{person}{Robert
  Kleinberg}.} \bibinfo{year}{2014}\natexlab{}.
\newblock \showarticletitle{Incentivizing exploration}. In
  \bibinfo{booktitle}{\emph{{ACM} Conference on Economics and Computation, {EC}
  '14, Stanford , CA, USA, June 8-12, 2014}}. \bibinfo{publisher}{ACM},
  \bibinfo{address}{New York, NY, USA}, \bibinfo{pages}{5--22}.
\newblock
\urldef\tempurl%
\url{https://doi.org/10.1145/2600057.2602897}
\showDOI{\tempurl}


\bibitem[\protect\citeauthoryear{Gittins, Glazebrook, and Weber}{Gittins
  et~al\mbox{.}}{2011}]%
        {Gittins-book11}
\bibfield{author}{\bibinfo{person}{John Gittins}, \bibinfo{person}{Kevin
  Glazebrook}, {and} \bibinfo{person}{Richard Weber}.}
  \bibinfo{year}{2011}\natexlab{}.
\newblock \bibinfo{booktitle}{\emph{{Multi-Armed Bandit Allocation Indices}}}.
\newblock \bibinfo{publisher}{John Wiley \& Sons}.
\newblock
\urldef\tempurl%
\url{https://doi.org/10.1002/9780470980033}
\showDOI{\tempurl}


\bibitem[\protect\citeauthoryear{Kamenica and Gentzkow}{Kamenica and
  Gentzkow}{2011}]%
        {Kamenica-aer11}
\bibfield{author}{\bibinfo{person}{Emir Kamenica} {and}
  \bibinfo{person}{Matthew Gentzkow}.} \bibinfo{year}{2011}\natexlab{}.
\newblock \showarticletitle{{Bayesian Persuasion}}.
\newblock \bibinfo{journal}{\emph{American Economic Review}}
  \bibinfo{volume}{101}, \bibinfo{number}{6} (\bibinfo{year}{2011}),
  \bibinfo{pages}{2590--2615}.
\newblock
\urldef\tempurl%
\url{https://doi.org/10.1257/aer.101.6.2590}
\showDOI{\tempurl}


\bibitem[\protect\citeauthoryear{Kremer, Mansour, and Perry}{Kremer
  et~al\mbox{.}}{2014}]%
        {Kremer-JPE14}
\bibfield{author}{\bibinfo{person}{Ilan Kremer}, \bibinfo{person}{Yishay
  Mansour}, {and} \bibinfo{person}{Motty Perry}.}
  \bibinfo{year}{2014}\natexlab{}.
\newblock \showarticletitle{Implementing the "Wisdom of the Crowd"}.
\newblock \bibinfo{journal}{\emph{J. of Political Economy}}
  \bibinfo{volume}{122} (\bibinfo{date}{Oct.} \bibinfo{year}{2014}),
  \bibinfo{pages}{988--1012}.
\newblock
Issue 5.
\urldef\tempurl%
\url{https://doi.org/10.1145/2482540.2482542}
\showDOI{\tempurl}
\newblock
\shownote{Preliminary version appeared in \emph{ACM Conf. on Economics and
  Computation}, 2014.}


\bibitem[\protect\citeauthoryear{Mansour, Slivkins, and Syrgkanis}{Mansour
  et~al\mbox{.}}{2015}]%
        {MansourSS15}
\bibfield{author}{\bibinfo{person}{Yishay Mansour}, \bibinfo{person}{Aleksandrs
  Slivkins}, {and} \bibinfo{person}{Vasilis Syrgkanis}.}
  \bibinfo{year}{2015}\natexlab{}.
\newblock \showarticletitle{Bayesian Incentive-Compatible Bandit Exploration}.
  In \bibinfo{booktitle}{\emph{Proceedings of the Sixteenth ACM Conference on
  Economics and Computation}} \emph{(\bibinfo{series}{EC '15})}.
  \bibinfo{publisher}{ACM}, \bibinfo{address}{New York, NY, USA},
  \bibinfo{pages}{565--582}.
\newblock
\showISBNx{978-1-4503-3410-5}
\urldef\tempurl%
\url{https://doi.org/10.1145/2764468.2764508}
\showDOI{\tempurl}


\bibitem[\protect\citeauthoryear{Mansour, Slivkins, Syrgkanis, and Wu}{Mansour
  et~al\mbox{.}}{2016}]%
        {MansourSSW16}
\bibfield{author}{\bibinfo{person}{Yishay Mansour}, \bibinfo{person}{Aleksandrs
  Slivkins}, \bibinfo{person}{Vasilis Syrgkanis}, {and}
  \bibinfo{person}{Zhiwei~Steven Wu}.} \bibinfo{year}{2016}\natexlab{}.
\newblock \showarticletitle{Bayesian Exploration: Incentivizing Exploration in
  Bayesian Games}. In \bibinfo{booktitle}{\emph{Proceedings of the 2016 ACM
  Conference on Economics and Computation}} \emph{(\bibinfo{series}{EC '16})}.
  \bibinfo{publisher}{ACM}, \bibinfo{address}{New York, NY, USA},
  \bibinfo{pages}{661--661}.
\newblock
\showISBNx{978-1-4503-3936-0}
\urldef\tempurl%
\url{https://doi.org/10.1145/2940716.2940755}
\showDOI{\tempurl}


\bibitem[\protect\citeauthoryear{Mansour, Slivkins, and Wu}{Mansour
  et~al\mbox{.}}{2018}]%
        {MansourSW18}
\bibfield{author}{\bibinfo{person}{Yishay Mansour}, \bibinfo{person}{Aleksandrs
  Slivkins}, {and} \bibinfo{person}{Zhiwei~Steven Wu}.}
  \bibinfo{year}{2018}\natexlab{}.
\newblock \showarticletitle{Competing Bandits: Learning Under Competition}. In
  \bibinfo{booktitle}{\emph{Innovations in Theoretical Computer Science
  Conference (ITCS 2018)}}. \bibinfo{pages}{48:1--48:27}.
\newblock
\showISBNx{978-3-95977-060-6}
\showISSN{1868-8969}


\bibitem[\protect\citeauthoryear{Slivkins}{Slivkins}{2017}]%
        {Slivkins17}
\bibfield{author}{\bibinfo{person}{Aleksandrs Slivkins}.}
  \bibinfo{year}{2017}\natexlab{}.
\newblock \showarticletitle{Incentivizing exploration via information
  asymmetry}.
\newblock \bibinfo{journal}{\emph{{ACM} Crossroads}} \bibinfo{volume}{24},
  \bibinfo{number}{1} (\bibinfo{year}{2017}), \bibinfo{pages}{38--41}.
\newblock
\urldef\tempurl%
\url{https://doi.org/10.1145/3123744}
\showDOI{\tempurl}


\end{thebibliography}

\appendix
\newpage
\section{Missing proofs from Sections 3, 4 and 5}
\label{app:sec3}

\begin{proof*}{Proof of Lemma \ref{BigThm}}
We prove here lemma for $q^2_t$ and $q_t^3$, i.e., $j=3$, and the full proof, for $j\geq 4$, can be found in Appendix \ref{appendix1}.
The values
$q_t^2$'s are stated in the theorem, and we will show that they satisfy the required conditions. In addition we state the values of $q^3_t$, and later show that they also satisfy the required conditions,
\begin{equation}
q^3_{t} = 
     \begin{cases}
       0 &\quad t<3\\
       \min(\frac{2p_{3}^{1}p_{1}^{-1}p_{2}^{-1}}{1-2p_{3}^{1}},p_2^{- 1}q_2^2) &\quad t=3\\
       \min 
(\frac{2p_{1}^{-1}p_{2}^{-1}p_{3}^{1}+p_{3}^{1}\sum_{\tau=3}^{t-1}{q^3_{\tau}}}{1-2p_{3}^{1}} 
,p_2^{-1}\sum_{\tau=2}^{t-1}{q^2_{\tau}}-\sum_{\tau=3}^{t-1}{q^3_{\tau}}) &\quad t>3\\
     \end{cases}
\end{equation}
Finally we need to show that $q_{t}^3\leq p_2^{-1}q_{t-1}^2$.

In the following, we calculate exploration rates for $\hat{\pi}$, and just for notational convenience, we will refer to $\Pr_{\hat{\pi}}[\ ]$ simply as $\Pr[\ ]$.
The first agent, by knowing her place in line, knows that none of the actions have been explored yet. Hence, it is for her best interest to choose action $1$, the action with the maximal prior expected reward. Therefore, it is necessary that $q_{1}^2=0$ for $\hat{\pi}$ to be BIC policy. Consequently, the first agent explores action $1$, and gains reward of $u_1(1)=x_1$.
As for the second agent, the planner can now use her knowledge of $x_1$ for $\sigma_2$.
If $x_1=1$, (the best possible reward), since she wishes to maximize the social welfare, she must recommend action $1$ to the rest of the agents, i.e., $S_t^{\langle 1, *,*\rangle}$ is a terminal state with maximal social welfare.
If $x_1=-1$, action 2 currently has the best expected reward, i.e., $\mu_2>\mu_3>- 1=x_1$. It also does not decrease the reward comparing to the known reward of action 1. Therefore, any social BIC planner (including $\hat{\pi}$) must recommend on action $2$ in this round, i.e., $\sigma_2=2$.
If $x_1=0$, although action $2$ is not the best action for agent $2$, a social planner would probably want to recommend it to the second agent, at least with some probability. For that she uses her advantage of knowing the realization of action $1$'s reward. Note that $q_2^2$ influences the second agent only when $x_1=0$.

To complete the definition of the recommendation for the second agent, we calculate the value of $q_2^2$. The expected utility of agent $2$ for following the recommendation must be at least the expected utility of choosing action $1$. Therefore, $q_2^2$ satisfies the BIC constraint, i.e.,
\[
\E[u_{2}(2)-u_{2}(1)|\sigma_2=2]\Pr[\sigma_2=2] = (\mu_2-(-1))\Pr[S_{2}^{\langle- 1,*,*\rangle}|\sigma_2=2]\Pr[
\sigma_2=2]+\mu_2q^2_{2}\geq 0.
\]
Information state $S_{2}^{\langle- 1,*,*\rangle}$ always leads the planner to recommend $\sigma_2=2$, therefore 
\[\Pr[S_{2}^{\langle- 1,*,*\rangle}|\sigma_2=2]\Pr[\sigma_2=2]=\Pr[S_{2}^{\langle- 1,*,*\rangle},\sigma_2=2]=\Pr[S_{2}^{\langle- 1,*,*\rangle}]
\]
According to Table \ref{RecTable}, the above is true for any action $j$ and exploitation state $S_{t}^{\vec{z}}\in \Gamma_t^{j+}$, which all have in common a corresponding recommendation $\sigma_t=j$. 
For this reason, for $S_{t}^{\vec{z}}\in \Gamma_t^{j+}$, we can replace $\Pr[S_{t}^{\vec{z}}|\sigma_t=j]\Pr[\sigma_t=j]$ by $\Pr[S_{t}^{\vec{z}}]$.

Recall that $q_2^2$ not only must satisfy the above constraint but also satisfy the second part of (\ref{qdefpart1}). Hence,
\[
(1+\mu_2)\Pr[S_{2}^{\langle - 1,*,*\rangle}]+\mu_2q^2_{2}\geq 0\ \quad \textnormal{and} \quad q_2^2\leq \Pr[S_{2}^{\langle0,*,*\rangle}]
\]
Since, $\mu_2=2p_2^1-1$, $\Pr[S_{2}^{\langle0,*,*\rangle}]=p_1^0$, and $\Pr[S_{2}^{\langle - 1,*,*\rangle}]=p_1^{- 1}$, we have
\[
q_2^2=\min(\frac{2p_2^1p_1^{-1}}{1-2p_2^1},p_1^0)\;.
\]
We proceed to calculate the first positive value of $q_t^3$. Due to the assumption $\mu_2>\mu_3$, agent 2 can deduce that action 3 has not been explored yet, consequently she would definitely not follow a recommendation to choose action 3, therefore, $q_{2}^3=0$.
Consider recommending action $3$ to agent $3$, i.e., $\sigma_3=3$.
This recommendation can occur when the planner is in either exploitation state $\langle- 1,- 1,*\rangle$, or exploration state $\langle 0,- 1,*\rangle$ . 
The BIC constraint is:
\[
\E[u_{3}(3)-u_{3}(1)|\sigma_3=3]\Pr[\sigma_t=3]=
2p_{3}^{1}\Pr[S_{3}^{\langle- 1,- 1,*\rangle}]+\mu_3 q^3_{3}\geq 0
\]
We again maximize over $q_3^3$, subject to the second constraint in (\ref{qdefpart1}) as well, i.e.,
\[
q_3^3\leq\frac{2p_{3}^{1}p_{1}^{-1}p_{2}^{-1}}{1-2p_{3}^{1}} \quad \textnormal{and} \quad q_3^3\leq \Pr[S_{3}^{\langle0,-1,*\rangle}]+\Pr[S_{3}^{\langle0,*,*\rangle}]=p_1^0-q_2^2p_2^1
\]
Since for $t=3$ and $j=3$, we assume that $q^2_2=A^2_2$, we have
$$q_3^3\leq A_3^3=\frac{2p_{3}^{1}p_{1}^{-1}p_{2}^{-1}}{1-2p_{3}^{1}}<\frac{2p_{2}^{1}p_{1}^{-1}p_{2}^{-1}}{(1-2p_{2}^{1})}=A_2^2p_2^{-1}=q_2^2p_2^{-1}=
q_2^2(1-p_2^{1})\leq p_1^0-q_2^2p_2^1$$
it implies that,
\begin{equation}\label{q3}
q_3^3=\min(\frac{2p_{3}^{1}p_{1}^{-1}p_{2}^{-1}}{1-2p_{3}^{1}},q_2^2p_2^{-1})
\end{equation}
Since $0<p_1^{-1},p_2^{-1},p_2^1,p_3^1<1$, and $\mu_3=(2p_3^1-1)<0$, we get that both the numerator and the denominator of $q_3^3$ are positive, 
therefore 
\begin{equation}\label{posq3}
q_3^3>0
\end{equation}

We now prove by induction on agent $t$, the following:
\begin{equation}\label{indClaim}
q_{t}^3\leq p_2^{-1}q_{t-1}^2
\end{equation}
This assumption assures that every time $\hat{\pi}$ recommends $\sigma_t=3$, it is done after action $2$ has been explored. (This will be clear after we define $f^2(y)$ and $f^3(y)$, however, observe that the exploration rate of the third action is bounded by the exploration of the second action up to the previous agent times the probability that the second action realization is $-1$.)
This will imply that $\Pr[S_t^{\langle0,*,-1\rangle}]=\Pr[S_t^{\langle0,*,1\rangle}]=\Pr[S_t^{\langle0,1,-1\rangle}]=0$.
This will simplify the derivation, as a recommendation $\sigma_t=2$ can come from only two exploitation states, $S_t^{\langle- 1, 1,*\rangle}, S_t^{\langle0, 1,*\rangle}$, or from $S_t^{\langle0, *,*\rangle}$, the only exploration driven state that may cause recommendation for action $2$.

For the induction base, consider agent $t=3$. Indeed, $q_3^3\leq q_2^2p_2^{-1}$ from (\ref{q3}).
For the induction step, assume that (\ref{indClaim}) holds for every time $t<t_0$.
The exploration rates, $q_t^2$ and $q_{t+1}^3$ for each agent $t=t_0$ can be derived from the following the constraints in (\ref{qdefpart1}).
Starting with $q_t^2$.
\begin{equation}\label{eq2}
\E[u_{t}(2)-u_{t}(1)|\sigma_t=2]\Pr[\sigma_t=2]=2\Pr[S^{\langle- 1,1,*\rangle}_t]+\mu_2q_{t}^2+\Pr[S_t^{\langle0,1,*\rangle}]\geq 0
\end{equation}
The probabilities of each mentioned state are as follows.
\begin{itemize}
\item $\Pr[S^{\langle- 1, 1,*\rangle}_t]=p_1^{-1}p_2^1$. Also, note that $S^{\langle- 1, 1,*\rangle}_t$ is a terminal state.
\item $\Pr[S_t^{\langle0, 1,*\rangle}]$ is the intersection of the following events:
\begin{itemize}
\item Action 2 has been explored before agent $t$, i.e. $\sum_{\tau=2}^{t-1}q_\tau^2$, so (implicitly) action 1 has already sampled and $x_1=0$.
\item $\Pr[X_2=1]=p_2^1$
\item This is a terminal state, so no further events.
\end{itemize}
\end{itemize}
Combining all with (\ref{eq2}),
\[
2p_{1}^{-1}p_{2}^{1}+(2p_{2}^{1}-1)q_{t}^2
+p_{2}^{1}\sum_{\tau=2}^{t-1}{q^2_{\tau}}\geq0
\]
and $q_t^2=\Pr_{\hat{\pi}}[\sigma_t=2, S_t^{\langle0,*,*\rangle}]\leq \Pr_{\hat{\pi}}[S_t^{\langle0,*,*\rangle}]\leq p_1^0-\sum_{\tau=2}^{t-1}{q_{\tau}^2}$. So, 
\[
q^2_t=min(\frac{2p_{1}^{-1}p_{2}^{1}+p_{2}^{1}\sum_{\tau=2}^{t-1}{q^2_{\tau}}}{1-2p_{2}^{1}},
p_1^0-\sum_{\tau=2}^{t-1}{q_{\tau}^2})
\]
As for $q_t^3$, the recommendation $\sigma_t=3$ can come from the exploitation states $S_t^{\langle- 1,- 1, 1\rangle}$ and $S_t^{\langle0,- 1,1\rangle}$ (recall that $\Pr[\sigma_3=3|S_3^{\langle- 1, - 1,*\rangle}]=1$, which implies that agent $3$ will perform action $3$, and therefore any agent $t\geq 4$ has $\Pr[S_t^{\langle- 1, - 1,*\rangle}]=0$), or from the only exploration state relevant for agent $t$, $S_t^{\langle0, - 1,*\rangle}$. Hence, the BIC constraint is
\[
\E[u_t(3)-u_t(1)|\sigma_t=3]\Pr[\sigma_t=3] =2\Pr[S_{t}^{\langle- 1,- 1,1\rangle}]+\mu_3 q_{t}^3+\Pr[S_{t}^{\langle0,- 1,1\rangle}]\geq 0
\]
substituting the probabilities we have
\[
2p_{1}^{-1}p_{2}^{-1}p_{3}^{1}+(2p_{3}^{1}-1)q^3
_{t}+p_{3}^{1}\sum_{\tau=3}^{t-1}{q^3
_{\tau}}=0
\]
Note that $\sum_{\tau=3}^{t-1}{q^3
_{\tau}}$ implicitly states that $x_1=0$ and that $x_2=-1$, or else action 3's reward might has been revealed by exploitation, or it still unknown, but did not revealed by exploration.
In order for $q^3_t$ to be a valid, it must also satisfy
\[
q_t^3=\Pr[\sigma_t=3,S_t^{\langle0,-1,*\rangle}]\leq \Pr[S_t^{\langle0,-1,*\rangle}]=p_2^{-1}\sum_{\tau=2}^{t-1}{q^2_{\tau}}-\sum_{\tau=3}^{t-1}{q^3_{\tau}}\;
\]
where $p_2^{-1}\sum_{\tau=2}^{t-1}{q^2_{\tau}}$ is the probability that there was an agent before agent $t$ that explored action 2, and $x_2=-1$. From this probability we subtract the probability that action 3 has been explored, i.e., $\sum_{\tau=3}^{t-1}{q^3_{\tau}}$, so that up until agent $t$ the state is $\langle0,-1,*\rangle$
Therefore,
\[
q^3_{t}=\min 
(\frac{2p_{1}^{-1}p_{2}^{-1}p_{3}^{1}+p_{3}^{1}\sum_{\tau=3}^{t-1}{q^3
_{\tau}}}{1-2p_{3}^{1}} 
,p_2^{-1}\sum_{\tau=2}^{t-1}{q^2_{\tau}}-\sum_{\tau=3}^{t-1}{q^3_{\tau}})
\]
We can upper bound $q_t^3$ as follows
\[
q^3_{t}\leq \frac{2p_{1}^{-1}p_{2}^{-1}p_{3}^{1}+p_{3}^{1}\sum_{\tau=3}^{t-1}{q^3
_{\tau}}}{1-2p_{3}^{1}}\leq
\frac{2p_{1}^{-1}p_{2}^{-1}p_{2}^{1}+p_{2}^{1}\sum_{\tau=3}^{t-1}{q^3_{\tau}}}{1-2p_{2}^{1}}
\]
the second inequality is correct due to the assumption that $p_3^1<p_2^1$.
From the induction hypothesis (\ref{indClaim}) we know that $q^3_{\tau}\leq p_2^{-1}q^2_{t-1}$ for every $\tau\leq t-1$. Since, for $t$ and $j=3$ we assume that $q^2_{t-1}=A^2_{t-1}$, we have,
\[
q_t^3\leq\frac{2p_{1}^{-1}p_{2}^{-1}p_{2}^{1}+p_{2}^{1}\sum_{\tau=3}^{t-1}{q^3_{\tau}}}{1-2p_{2}^{1}}\leq 
\frac{2p_{1}^{-1}p_{2}^{-1}p_{2}^{1}+p_{2}^{1}p_2^{-1}\sum_{\tau=2}^{t-2}{q^2
_{\tau}}}{1-2p_{2}^{1}}= 
p_2^{-1}A_{t-1}^2=p_2^{-1}q_{t-1}^2\;.
\]
By using the induction hypothesis again, with (\ref{posq3}), we get $0<\sum_{\tau=3}^{t-1}{q^3_{\tau}}<p_2^{-1}\sum_{\tau=2}^{t-2}{q^2_{\tau}}$, thus
\[
q_t^3\leq p_2^{-1}q_{t-1}^2\leq p_2^{-1}q_{t-1}^2+ p_2^{-1}\sum_{\tau=2}^{t-2}{q^2_{\tau}}-\sum_{\tau=3}^{t-1}{q^3_{\tau}}=p_2^{-1}\sum_{\tau=2}^{t-1}{q^2_{\tau}}-\sum_{\tau=3}^{t-1}{q^3_{\tau}}
\]
which completes the proof of (\ref{indClaim}) and the proof of the theorem.
\end{proof*}

\begin{proof*}{Proof of Lemma \ref{qAfternj}}

First, we show that if action $j$ is recommended for agent $t>n_{j-1}$, then action $(j-1)$'s reward has been observed. If $t>n_{j-1}$ then $q_t^{j-1}=0$. From the constraints in (\ref{qdefpart1}) regarding action $j-1$ for $t>n_{j-1}\geq j$, as the ``gain'' part is strictly positive and the ``loss'' part is strictly negative, we get that $\sum_{\Gamma_t^{(j-1)-}}\Pr_{\hat{\pi}}[S_t^{\vec{z}}]=0$.
Therefore, all the states for which action $j$ is explored before action $j-1$ are infeasible for agent $t$ as they were for $t\leq n_{j-1}$ and we get to continue the induction without relying on $q_{t}^{j}\leq p_{j-1}^{-1}A_{t-1}^{j-1}$. Hence, we get exactly the same constraints for action $j$ as for the case of $t\leq n_{j-1}$, i.e.,
\[
\E[u_t(j)-u_t(1)|\sigma_t=j]\Pr[\sigma_t=j] =
2p_j^1\Pi_{i<j}p_{i}^{-1}+(2p_{j}^{1}-1)q^j
_{t}+p_{j}^{1}\sum_{\tau=j}^{t-1}{q^j
_{\tau}}\geq 0
\]
and also,
\[
q_t^j\leq p_{j-1}^{- 1}\sum_{\tau=j-1}^{t-1}{q^{j-1}_{\tau}}-\sum_{\tau=j}^{t-1}{q^j_{\tau}}
\]
which completes the proof of the lemma.
\end{proof*}

\begin{proof*}{proof of Lemma \ref{APositive}}
The proof is by induction over $t$.
For the induction base consider $t=j$.
From Lemma \ref{corq3bigger0}, we have that $A_j^j\geq q_j^j>0$, therefore
\[ 
A_{j}^j =\frac{2p_{j}^{1}\Pi_{i<j}p_{i}^{-1}}{1-2p_{j}^{1}}<
\frac{2p_{j}^{1}\Pi_{i<j}p_{i}^{-1}+p_{j}^{1}{q^j_{j}}}{1-2p_{j}^{1}}=A_{j+1}^j
\]
For the inductions step, we assume that the induction hypothesis holds for $t$ and prove it for $t+1$.
From the inductive hypothesis we have that $A_{t-1}^j< A_{t}^j$. If $q^j_{t}>0$, we get
\[ 
A_{t}^j =\frac{2p_{j}^{1}\Pi_{i<j}p_{i}^{-1}+
p_{j}^{1}\sum_{\tau=j}^{t-1}{q^j_{\tau}}}{1-2p_{j}^{1}}<
\frac{2p_{j}^{1}\Pi_{i<j}p_{i}^{-1}+p_{j}^{1}\sum_{\tau=j}^{t}{q^j_{\tau}}}{1-2p_{j}^{1}}=A_{t+1}^j
\]
which proves the lemma.
\end{proof*}

\begin{proof*}{Proof of Lemma~\ref{BigThm2}}
We show this lemma by induction over action index $j$. 
For the base of the induction, consider action $j=2$.
We prove the case of action $2$ by induction over $t$. For this induction we use the base case of $t=n_1+1=2$ (since agent $1$ always explores action $1$). We have
\begin{enumerate}
\item $B_2^2\geq q_2^2 > 0=B_1^2$.
\item $0<q_2^2$, directly from Lemma \ref{qjPositive}.
\item $q_2^1=0$ .
\item If $q_2^2=B_2^2=p_1^0(>0)$, we show by induction that for every $i\geq1$ it holds that $q_{2+i}^2=B_{2+i}^2=0$. For base consider $i=1$. Then by using Lemma \ref{APositive} we get
\[q_{3}^2=min(A_{3}^2,B_{3}^2)=min(A_{3}^2,p_1^0-q_2^2)=min(A_{3}^2,p_1^0-p_1^0)=min(A_3^2,0)=0\;.\]
For the induction step assume this property holds for $i-1$ and show for $i$. From the induction hypothesis and Lemma \ref{APositive} we get
\[q_{2+i}^2=min(A_{2+i}^2,B_{2+i}^2)=min(A_{2+i}^2,B_{2+i-1}^2-q_{2+i-1}^t)=min(A_{2+i}^2,0)=0\]
\end{enumerate}
For the induction step we assume that hypothesis of the lemma holds for action $2$ and for every agent $\leq t-1$, and show it for $t$.
\begin{enumerate}
\item From the induction hypothesis $q_{t-1}^2> 0$, so we get
\[B_{t-1}^2=p_1^0-\sum_{m=2}^{t-2}{q_{m}^2} > p_1^0-\sum_{m=2}^{t-1}{q_{m}^2}= B_t^2\]
\item Since $B_{t-1}^2 > B_t^2$ and that $q_{t-1}^2=A_{t-1}^2$ (or else $t>n_2$), and we know that $A_{t-1}^2>0$ from Lemma \ref{APositive}. Therefore,
\[q_t^j=\min(A_t^j,B_t^2)<\min(A_{t}^j,B_{t-1}^2)=\min(A_{t}^j,q_{t-1}^2)>0\]
\item If $q_t^j=B_t^2(>0)$, we show by induction that for every $i\geq1$ it holds that $q_{t+i}^2=B_{t+i}^2=0$. For base consider $i=1$. Then by using Lemma \ref{APositive} we get
\[q_{t+1}^2=min(A_{t+1}^2,B_{t+1}^2)=min(A_{t+1}^2,B_t^2-q_t^2)=min(A_{t}^2,0)=min(A_t^2,0)=0\;.\]
For the induction step consider that this property holds for $i-1$ and show for $i$. From the induction hypothesis and Lemma \ref{APositive} we get
\[q_{t+i}^2=min(A_{t+i}^2,B_{t+i}^2)=min(A_{t+i}^2,B_{t+i-1}^2-q_{t+i-1}^2)=min(A_{t+i}^2,0)=0\]
\end{enumerate}
So the hypothesis of the lemma holds for action $j=2$.

We now assume that it holds for every action $\leq j-1$ and show it for action $j$, again by induction over $t$.
For base, consider $t=j$
\begin{enumerate}
\item $B_j^j\geq q_j^j > 0=B_{j-1}^j$.
\item $0<q_j^j$, directly from Lemma \ref{qjPositive}.
\item If $q_j^j=B_j^j(>0)$, we show by induction that for every $i\geq1$ it holds that $q_{j+i}^j=B_{j+i}^j=0$. For the induction base consider $i=1$. Then by using Lemma \ref{APositive} and the since we assume that $t=j>n_{j-1}$, it holds that $q^{j-1}_{j}=0$ therefore 
\[q_{j+1}^j=min(A_{j+1}^j,B_{j+1}^j)=min(A_{j+1}^j,B_j^j+p_{j-1}^{- 1}{q^{j-1}_{j}}-q_{j}^j)=min(A_{j+1}^j,p_{j-1}^{- 1}{q^{j-1}_{j}})=min(A_{j+1}^j,0)=0\;.\]
For the induction step consider that this property holds for $i-1$. Since $t=j>n_{j-1}$, it holds that $q^{j-1}_{j+i-1}=0$. Using Lemma \ref{APositive}, we get
\[q_{j+i}^j=min(A_{j+i}^j,B_{j+i}^j)=min(A_{j+i}^j,B_{j+i-1}^j-q_{j+i-1}^j+p_{j-1}^{- 1}{q_{j-1+i}^{j}})=min(A_{j+i}^j,0)=0\]
\end{enumerate}
For the induction step we assume that the hypothesis of the lemma holds for agent $ t-1$ and for both actions $j-1$ and $j$. We now show that it holds for agent $t$ with action $j$.
\begin{enumerate}
\item Since $t>n_{j-1}$ we get
\[\sum_{\tau=j-1}^{t}{q^{j-1}_{\tau}}=\sum_{\tau=j-1}^{t-1}{q^{j-1}_{\tau}}=\sum_{\tau=j-1}^{n_{j-1}}{q^{j-1}_{\tau}}\]
And from the induction hypothesis $0<q_{t-1}^j$, hence
\[B_t^j=p_{j-1}^{- 1}\sum_{\tau=j-1}^{n_{j-1}}{q^{j-1}_{\tau}}-\sum_{\tau=j}^{t-1}{q^j_{\tau}}>
p_{j-1}^{- 1}\sum_{\tau=j-1}^{n_{j-1}}{q^{j-1}_{\tau}}-\sum_{\tau=j}^{t-2}{q^j_{\tau}}=B_{t-1}^j\]
\item From the induction hypothesis and $t\leq n_j$ we get $0<B_t^j$.
From Lemma \ref{APositive}, $0<A_t^j$, therefore $0<q_t^j$. 
\item If $q_t^j=B_t^j(>0)$, we show by induction that for every $i\geq1$ it holds that $q_{t+i}^j=B_{t+i}^j=0$. 
For the induction base consider $i=1$. Then by using Lemma \ref{APositive} and the since $t>n_{j-1}$, it holds that $q^{j-1}_{t}=0$,therefore,
\[q_{t+1}^j=min(A_{t+1}^j,B_{t+1}^j)=min(A_{t+1}^j,B_j^j+p_{j-1}^{- 1}{q^{t-1}_{j}}-q_{t}^j)=min(A_{t+1}^j,p_{j-1}^{- 1}{q^{t-1}_{j}})=min(A_{t+1}^j,0)=0\;.\]
For the induction step consider that this property holds for $i-1$. Since $t=j>n_{j-1}$, it holds that $q^{j-1}_{t+i-1}=0$. Using Lemma \ref{APositive}, we get
\[q_{t+i}^j=min(A_{t+i}^j,B_{t+i}^j)=min(A_{t+i}^j,B_{t+i-1}^j-q_{t+i-1}^j+p_{j-1}^{- 1}{q_{t-1+i}^{j}})=min(A_{t+i}^j,0)=0\]
\end{enumerate}
\end{proof*}

\begin{proof*}{Proof of Lemma~\ref{lemmaRho}}

For every $t_1<t_2$, it holds that $[\sigma_{t_2}=j, S_{t_2}^{\vec{z}}\in \Gamma^{j-}_{t_2}]\cap[\sigma_{t_1}=j, S_{t_1}^{\vec{z}}\in \Gamma^{j-}_{t_1}]=\emptyset$, since if $\sigma_{t_1}=j$ then $S_{t_2}^{\vec{z}}\notin \Gamma^{j-}_{t_2}$. From Theorem \ref{qBeforeLast}, it implies that
\[
\Pr[\exists t: \sigma_t=j, S_t^{\vec{z}}\in \Gamma^{j-}_t]=\sum_{t=1}^{T}\Pr[\sigma_t=j, S_t^{\vec{z}}\in \Gamma^{j-}_t]=\sum_{t=j}^{n_j}{q^{j}_{t}}=B_{n_j}^j+\sum_{t=j}^{n_j-1}{q^{j}_{t}}
\]
We now prove by induction on action $j$ the following:
\[\Pr[\exists t: \sigma_t=j, S_t^{\vec{z}}\in \Gamma^{j-}_t]=\rho_j\]
The induction base is done for action $j=2$ as follows
\[
\Pr[\exists t: \sigma_t=2, S_t^{\vec{z}}\in \Gamma^{2-}_t]=q_{n_2}^2+\Sigma_{t=1}^{n_2-1}q_t^2=p_1^0-\sum_{\tau=2}^{n_2-1}{q^2_{\tau}}+\sum_{t=1}^{n_2-1}q_t^2=p_1^0
\]
Suppose the induction hypothesis is true for action $j-1$. For action $j$ we have,
\[
\Pr[\exists t: \sigma_t=j, S_t^{\vec{z}}\in \Gamma^{j-}_t]=q_{n_j}^j+\Sigma_{t=1}^{n_j-1}q_t^j=p_{j-1}^{-1}\sum_{\tau=j-1}^{n_j-1}{q^{j-1}_{\tau}}-\sum_{\tau=j}^{n_j-1}{q^j_{\tau}}+\sum_{t=1}^{n_j-1}q_t^j=p_{j-1}^{-1}\sum_{\tau=j-1}^{n_j-1}{q^{j-1}_{\tau}}
\]
From the induction hypothesis we have that 
\[
\Pr[\exists t: \sigma_t=j-1, S_t^{\vec{z}}\in \Gamma^{(j-1)-}_t]=\sum_{t=j-1}^{n_{j-1}}{q^{j-1}_{t}}=\rho_{j-1}
\]
We get 
\[
\Pr[\exists t: \sigma_t=j, S_t^{\vec{z}}\in \Gamma^{j-}_t]=p_{j-1}^{-1}\sum_{\tau=j-1}^{n_j-1}{q^{j-1}_{\tau}}=p_{j-1}^{-1}\rho_{j-1}=\rho_j
\]
which completes the proof of the lemma.
\end{proof*}

\begin{proof*}{Proof of Lemma~\ref{f2smallest}}
From the definitions of $f^j(y)$, we get
\[
f^j(y)+1= \textnormal{argmax}_{t+1}(\sum_{\tau=j}^{t-1}q_{\tau}^{j}<y\rho_j)=\textnormal{argmax}_{t+1}(\sum_{\tau=j}^{t-1}p_{j}^{-1}q_{\tau}^{j}<yp_{j}^{-1}\rho_j)
\]
Since $q_{\tau+1}^{j+1}\leq p_j^{-1}q_\tau^j$ for every agent $\tau<t\leq n_{j}$ from Lemma \ref{BigThm}, and $p_{j}^{-1}\rho_j=\rho_{j+1}$, we get
\[
f^j(y)<f^j(y)+1\leq \textnormal{argmax}_{t+1}(\sum_{\tau=j}^{t-1}q_{\tau+1}^{j+1}<y\rho_{j+1})=
\textnormal{argmax}_{t+1}(\sum_{\tau=j+1}^{t}q_{\tau}^{j+1}<y\rho_{j+1})=
\]
\[
\textnormal{argmax}_{t}(\sum_{\tau=j+1}^{t-1}q_{\tau}^{j+1}<y\rho_{j+1})=f^{j+1}(y)
\]
hence $f^{j}(y)<f^{j+1}(y)$ for every $y\in[0,1]$.
\end{proof*}

\begin{proof*}[Proof of Lemma~\ref{welldefined}]
The first part, $f^i(y) \ne f^{j}(y)$, is direct consequence of Lemma \ref{f2smallest}. As for the second part, for every action $j$ and each agent $t$, $q_t^j\geq 0$ and $\sum_{t=1}^{n_j}q_t^j=\rho_j$, by Lemma \ref{lemmaRho}.
We also have monotone increasing exploration between agents- for every $j\leq t < n_j-1$ and for every $j\in\{2,...,k-1\}$, $q^j_t<q^{j}_{t+1}$, as a result of Lemma \ref{APositive} ($A^j_t<A^{j}_{t+1}$) and Lemma \ref{qBeforeLast} ($q^j_t=A^j_t$).

Let $y\in[0,1]$. Then $f^j(y)=\textnormal{argmax}_{t}(\sum_{\tau=1}^{t-1}q_{\tau}^j<y\rho_j)$. From the above get 
\[
0=q_1^j\leq \sum_{\tau=1}^{t-1}q^j_\tau\leq \sum_{\tau=1}^{n_{j}}q^j_\tau= \rho_j
\]
which completes the proof.
\end{proof*}

\begin{proof*}{Proof of Lemma~\ref{stoDom}}

For action $j\in \{2,\dots,k\}$ let $(\psi^j)_{j=1}^k$ denote the exploration rates used in $\pi_A$, i.e., \\$\psi^j_t=\sum_{S_t^{\vec{z}}\in{\Gamma_t^{j-}}}\Pr_{\pi_A}[S_t^{\vec{z}},\sigma_t=j]$.
Both $\pi_A$ and $\hat{\pi}$ are BIC policy algorithms, with the same recommendations for exploitation states. This implies that the only difference between the probabilities of $\pi_A$ and $\hat{\pi}$ to know action $j$'s reward at time $t$ (i.e., the probability $S^{\vec{z}}_t[j]\neq*$) is the difference between the sum of exploration rates of action $j$ until time $t$. Meaning that for every BIC policy $\pi$ with the same recommendations for exploitation states like $\hat{\pi}$,
\[
\Pr_{\pi}[S_t^{\vec{z}}\wedge \vec{z}[j]\ne *]=
\sum_{\tau=1}^{t-1}\sum_{S_\tau^{\vec{z}}\in{\Gamma_\tau^{j-}}}
\Pr_{\pi}[S_\tau^{\vec{z}},\sigma_\tau=j]+
\Pr_{\pi}[\forall i<j: x_i=- 1, j<t]
\]
The sum
$\sum_{\tau=1}^{t-1}
\sum_{S_\tau^{\vec{z}}\in{\Gamma_\tau^{j-}}}
\Pr_{\pi}[S_\tau^{\vec{z}},\sigma_\tau=j]$ is the sum of all exploration rates of action $j$ until $t-1$, therefore
\[
\Pr_{\hat{\pi}}[S_t^{\vec{z}}\wedge \vec{z}[j]\ne *]- \Pr_{\pi_A}[S_t^{\vec{z}}\wedge \vec{z}[j]\ne *]=\Sigma_{\tau=j}^{t-1}{q_\tau^j} -\Sigma_{\tau=j}^{t-1}{\psi_{\tau}^j}
\]
for every agent $t$.

For contradiction, suppose that there exists a prior, $D$, an action $j$ and time $t_1$ such that
\begin{equation}\label{contr}
\Pr_{\hat{\pi}}[S_{t_1}^{\vec{z}}\wedge \vec{z}[j]\ne *]<\Pr_{\pi_A}[S_{t_1}^{\vec{z}}\wedge \vec{z}[j]\ne *]\;,
\end{equation}
where $j$ is the least such action and $t_1$ is the least such agent for action $j$.

Every $q_t^j$ was calculated inductively so that it would attain a maximum value and maintain the constraints in (\ref{qdefpart1}), independent of the other actions. Therefore it is not possible that there exists time $t_0<t_1$ such that $\psi_{t}^j=q_{t}^j$ for every $t<t_0$ and $\psi_{t_0}^j>q_{t_0}^j$.
If $\psi_{t}^j=q_{t}^j$ for every $t<t_1$, then $\psi_{t_1}^j\leq q_{t_1}^j$.

So let $t_0<t_1$ denote the first time that there is lower exploration rate in $\pi_A$ for action $j$ rather than in $\hat{\pi}$, i.e., $t_0=\textnormal{argmin}_{t<t_1}\psi_t^j<q_t^j$ and $\psi_{t}^j=q_{t}^j$ for every $t<t_0$. Hence, the probability that $\pi_A$ is in exploitation state at time $t=t_0+1$ w.r.t. $\hat{\pi}$ is lower, i.e.,
\[
\sum_{S_t^{\vec{z}}\in \Gamma_t^{j+}}\Pr_{\pi_A}[S_t^{\vec{z}}|\sigma_t=j] =
\sum_{S_t^{\vec{z}}\in \Gamma_t^{j+}}\Pr_{\pi_A}[S_t^{\vec{z}}] <
\sum_{S_t^{\vec{z}}\in \Gamma_t^{j+}}\Pr_{\hat{\pi}}[S_t^{\vec{z}}]=
\sum_{S_t^{\vec{z}}\in \Gamma_t^{j+}}\Pr_{\hat{\pi}}[S_t^{\vec{z}}|\sigma_t=j]\;,
\]
and the probability that $\pi_A$ is in exploration state at time $t=t_0+1$ w.r.t. $\hat{\pi}$ is higher, i.e.,
\[
\sum_{S_t^{\vec{z}}\in \Gamma_t^{j-}}\Pr_{\hat{\pi}}[S_t^{\vec{z}}] < \sum_{S_t^{\vec{z}}\in \Gamma_t^{j-}}\Pr_{\pi_A}[S_t^{\vec{z}}]\;.
\]
Since $\E[u_t(j)-u_t(1)|S^{\vec{z}}_t,\sigma_t=j]$ and $\mu_j$ depend only on the prior $D$, their value remain the same. As a result of $\psi_{t_0}^j<q_t^j$, the exploration rate of action $j$ in $\pi_A$ at the next time, $t_0+1$ that maintains the constraint is smaller than it could have been while using $\hat{\pi}$, and for every time $t>t_0$ it holds that $\psi_{t}^j\leq q_t^j$. This is true for every time $t_0<t_1$ such that $\psi_{t_0}^j<q_{t_0}^j$ and therefore contradicts (\ref{contr}). Therefore, $\psi_t^j\leq q_t^j$ for every action $j$ and agent $t$.

The difference between the policies indicates that there is an action $j$ agent $t$ with $\psi_t^j\neq q_t^j$. Since $\psi_t^j\leq q_t^j$, 
it implies that $\psi_t^j < q_t^j$ and we get
\begin{equation}
\Pr_{\hat{\pi}}[S_{t+1}^{\vec{z}}\wedge \vec{z}[j]\ne *]<\Pr_{\pi_A}[S_{t+1}^{\vec{z}}\wedge \vec{z}[j]\ne *]\;,
\end{equation}
Which completes the proof.
\end{proof*}

\begin{proof*}{Proof of Theorem~\ref{zeroReg}}
For the sake of contradiction, assume that there exists a prior $D$, such that
\[
SW_{T}(OPT)=\E_D[\Sigma_{t=1}^T u_t(\pi_{opt}(h_{t-1}))]>\E_D[\Sigma_{t=1}^T u_t(\hat{\pi}(h_{t-1}))]=SW_{T}(\hat{\pi})
\]
$\pi_{opt}$ maximizes expected social welfare, therefore it is easy to see that $\pi_{opt}$ must give the same recommendation as in Table ~\ref{RecTable} whenever that social planner is in exploitation state. Since $\pi_{opt}$ and $\hat{\pi}$ are different,there is at least one agent $t$, that might not receive the same recommendation from the two policies i.e., $(v_t)_{\pi_{opt}}\ne(v_t)_{\hat{\pi}}$. We have already established that this scenario can only happen when the planner is in exploration state. Therefore, it is a result of difference between exploration rates in both policies for at least one action $j$. Let $t_0$ and $j$ denote the first indexes of such agent and action, respectively. It is easy to see that if $t_0>n_j$ then $\pi_{opt}$ is not in a terminal state and therefore does not maximizes social welfare. So assume $t_0\leq n_j$. Let $(\psi_{t}^{j})_{\pi_{opt}}$ denote the exploration rate used by $\pi_{opt}$, and let $\epsilon$ denote the difference between this exploration rate, and in $\hat{\pi}$, i.e., $\epsilon=q_{t_0}^{j}-(\psi_{t_0}^{j})_{\pi_{opt}}$. Let $\pi$ be a policy identical to $\pi_{opt}$ that substitutes $(\psi_{t_0}^{j})_{\pi_{opt}}$ with $q_{t_0}^{j}$. Recommendation policy $\pi$ is a well defined BIC policy as $t_0$ is the first index of $j$ for which $(\psi_{t_0}^{j})_{\pi_{opt}}$ is not the maximum value, $q_{t_0}^j$ is a BIC exploration rate, and from Lemma \ref{stoDom}, the rest of the exploration rates can still be used. 

Let $t_1$ be the agent such that 
$t_1=\lceil{t_0+\frac{1}{p_{j}^1}-1}\rceil$. Since $t_0\leq n_j\leq n_k$ and $p_j^1\geq p_k^1$, it holds that $t_1\leq\lceil{n_k+\frac{1}{p_{k}^1}-1}\rceil=T$.
Now, since 
\[
SW_{T}(\pi)-SW_{T}(\pi_{opt})\geq \E_D[\Sigma_{t=t_0}^{t_1} u_t(\pi_{opt}(h_{t-1}))]-\E_D[\Sigma_{t=t_0}^{t_1} u_t(\pi(h_{t-1}))]
\]
We get
\[
SW_{T}(\pi)-SW_{T}(\pi_{opt})\geq \epsilon ((t_1-t_0)p_j^1+2p_j^1-1)=\epsilon((\frac{1}{p_{j}^1}-1)p_j^1+2p_j^1-1)=\epsilon p_j^1
\]
contradicting the optimality of $\pi_{opt}$.
\end{proof*}

\begin{proof*}{Proof of Lemma \ref{monoincI}}
Let $j\ne 1$ be some action. The proof is by induction over $t$.\\
For the induction base, consider $t\leq j-1$, for which $i_{t}^{j}=i_{t+1}^j=-1$.\\
For the induction step, we assume the induction hypothesis holds for $t$ and prove it for $t+1$.\\
The left-hand sides of (\ref{defijj}) and (\ref{defitj}) are non-negative since 
$i_{t}^{j}\leq 1$, $X_1<1$, $0<p_n^1$ and $p_j^1>0$. Consequently, the right-hand sides are also non-negative. So from (\ref{defijj}) we deduce $-1<\mu_j<i_{j+1}^j$ and from (\ref{defitj}) and the induction hypothesis we deduce $\omega_{t+1}^j\geq i_{t}^j$ and as a result, $i_{t+1}^j\geq i_{t}^j$.
\end{proof*}

\begin{proof*}{proof of Lemma \ref{incI}}
Let $j\ne 1$ be some action.The proof is by induction over $t$.\\
For the base case, consider $t=j$, for which $i^{j}_j=\mu_j>\mu_{j+1}=i^{j+1}_{j+1}$.\\
For the induction step, we assume the induction hypothesis holds for every agent $\leq t$.\\
Consider $t=j$, then according to (\ref{defijj}) $\omega_{j+1}^j$ is the solution to:
\begin{equation}\label{simplejj}
\prod_{n=2}^{j-1}(p_n^{-1})\int_{\mu_j\geq X_1} [\mu_j-X_1]dD_1=
\int_{\mu_{j}\leq X_1\leq \omega_{j+1}^j} [X_1-\mu_j]dD_1 
\end{equation}
Now, since $0<p_j^{-1}<1$ and $\mu_j>\mu_{j+1}$,
\[
\prod_{n=2}^{j-1}(p_n^{-1})\int_{\mu_j\geq X_1} [\mu_j-X_1]dD_1>\prod_{n=2}^{j}(p_n^{-1})\int_{\mu_{j+1}\geq X_1} [\mu_{j+1}-X_1]dD_1
\]
Putting the above inequality with (\ref{simplejj}) for both $j$ and $j+1$ we get,
\[
\int_{\mu_j\leq X_1\leq \omega_{t+1}^j} [X_1-\mu_j]dD_1 >
\int_{\mu_{j+1}\leq X_1\leq \omega_{t+2}^{j+1}} [X_1-\mu_{j+1}]dD_1 
\]
Since $\mu_j>\mu_{j+1}$ (and therefore $-\mu_j<-\mu_{j+1}$) $\omega_{t+2}^{j+1}<\omega_{t+1}^{j}$ must hold.

\item For every $t>j$, according to (\ref{defitj}) $\omega_{t+1}^j$ is the solution to:
\begin{equation}\label{simpletj}
p_j^1\int_{-1\leq X_1\leq i_{t}^j} [1-X_1]dD_1=
\int_{i_{t}^j< X_1\leq \omega_{t+1}^j} [X_1-\mu_j]dD_1 
\end{equation}

Combining the induction hypothesis ($i_{t}^j>i_{t+1}^{j+1}$) with $p_{j+1}^{1}<p_{j}^{1}$ we get
\[
p_j^1\int_{-1\leq X_1\leq i_{t}^j} [1-X_1]dD_1>
p_{j+1}^1\int_{-1\leq X_1\leq i_{t+1}^{j+1}} [1-X_1]dD_1
\]
Putting the above inequality with (\ref{simpletj}) for both $j$ and $j+1$ we get,
\[
\int_{i_{t}^j< X_1\leq \omega_{t+1}^j} [X_1-\mu_j]dD_1 > 
\int_{i_{t+1}^{j+1}<X_1\leq \omega_{t+2}^{j+1}} [X_1-\mu_{j+1}]dD_1 
\]
Now, since $-\mu_j<-\mu_{j+1}$ and we know that $i_{t}^j>i_{t+1}^{j+1}$ from the induction hypothesis, we get that $\omega_{t+2}^{j+1}<\omega_{t+1}^{j}$.
Consequently, for every $t\geq$ it holds that $i_{t+2}^{j+1}\leq i_{t+1}^{j}$.
\end{proof*}

\begin{proof*}{Proof of Lemma \ref{stoDomcont}}
For contradiction, suppose that there exists a prior, $D$, an action $j$ and time $t_1$ such that
\begin{equation}
\Pr_{\hat{\pi}}[S_{t_1}^{\vec{z}}\wedge \vec{z}[j]\ne *|\vec{z}[1]=x_1]<\Pr_{\pi_A}[S_{t_1}^{\vec{z}}\wedge \vec{z}[j]\ne *|\vec{z}[1]=x_1]\;,
\end{equation}
where $j$ is the least such action and $t_1$ is the least such agent for action $j$.
It means that $\pi_A$ recommends to some agent $j\leq t<t_1$ to explore action $j$ while $\vec{z}[j]=*$ and $(\sigma_t)_{\hat{\pi}}\ne j$. If $(\sigma_t)_{\hat{\pi}}= 1$ it means that $\theta^t_j=\emptyset$, therefore from Lemma \ref{BICisPartition}, $\pi_A$ is not a BIC policy (as it does not a valid partition policy). So consider exploration driven recommendation (i.e., $S_{t}^{\vec{z}}\wedge \vec{z}[i]= *\wedge \vec{z}[j]= *$), $(\sigma_t)_{\hat{\pi}}= i< j$ (since $i>j$ contradictions Corollary \ref{ActionsbyOrder}). In such a case,
\begin{equation}
\Pr_{\hat{\pi}}[S_{t+1}^{\vec{z}}\wedge \vec{z}[i]\ne *|\vec{z}[1]=x_1]>\Pr_{\pi_A}[S_{t+1}^{\vec{z}}\wedge \vec{z}[i]\ne *|\vec{z}[1]=x_1]\;,
\end{equation}
Hence $\pi_A$ is not stochastic dominant over $\hat{\pi}$.
\end{proof*}

\begin{proof*}{Proof for Theorem \ref{swCont}}
For the sake of contradiction, assume that there exists a prior $D$ and a realization $X_1=x_1$, such that
\[
SW_{T}(OPT)=\E_D[\Sigma_{t=1}^T u_t(\pi_{opt}(h_{t-1}))]>\E_D[\Sigma_{t=1}^T u_t(\hat{\pi}(h_{t-1}))]=SW_{T}(\hat{\pi})
\]
$\pi_{opt}$ maximizes expected social welfare, so according to Lemma \ref{BICisPartition} it must be a partition policy. By differing from $\hat{\pi}$, there exist an action $j$ and a time $t_1$ in which $(\sigma_{t_1})_{\hat{\pi}}=j$ and $(\sigma_{t_1})_{\pi_{OPT}}=i\ne j$, where $j$ is the least such action and $t_1$ is the least such agent for action $j$. 
If $i<j$ then action $i$ is already explored by both planners and $x_i=-1$, and the result is a lower $SW$ for OPT.
If $i>j$, since $\mu_i<\mu_j$ the result is a lower social welfare for OPT.
Overall we get that $SW_{T}(OPT)<SW_{T}(\hat{\pi})$.
\end{proof*}
\section{K actions- Full Algorithm and Missing Proofs}\label{appendix1}
\begin{table}
    
	\begin{minipage}{\columnwidth}
		\begin{center}
          \begin{tabular}{l|l|l|l|l}
          	\toprule
$X_1$ &$X_j$& $\sigma_t$ & 
\textit{$\E[u_t(j)-u_t(1)|S^{\vec{z}}_t,\sigma_t=j]$} & $\Pr[S_t^{\vec{z}}]$\\ \hline
$-1$ & $*$&$j$ &$2p_j^1$&$\mathbbm{1}[t=j]\Pi_{i<j}p_{i}^{-1}$\\ \hline
$0$ &$1$& $j$ &1&$p_{j}^{1} \sum_{\tau=j}^{t-1}{q^j_{\tau}}$\\
\hline
$-1$ & $1$ & j & $2$ &$\mathbbm{1}[t=j]p_{j}^{1}\Pi_{i<j}p_{i}^{-1}$\\
\hline
$0$ & $*$ & $f^j_t(y)\in \{1,j\}$ & $2p_j^1-1$ &$p_{j-1}^{- 1}\sum_{\tau=j-1}^{t-1}{q^{j-1}_{\tau}}-\sum_{\tau=j}^{t-1}{q^j_{\tau}}$\\
\hline
				\bottomrule
			\end{tabular}
		\end{center}
		\bigskip\centering
	\end{minipage}
		\caption{Extension for states that may recommend on action $j$}\label{tblExtj}
\end{table}%
In this appendix, we extend the algorithm for 3 actions to handle with any number of actions, k.
Recall that when calculating BIC constraints, we consider 3 different reasons for a recommendation, $\sigma_t=j$, as explained in Example \ref{expstates}. To handle with multiple actions' rewards, we abandon the explicit states, e.g., $S_t^{\langle- 1,1,*\rangle}$.
We intentionally dismiss any states where actions are, as we will see that these states are infeasible by $\hat{\pi}$.

\begin{enumerate}
\item \textbf{Exploitation driven recommendation}, action $j$ either has:
\begin{enumerate}
\item \textbf{A known reward}- The planner has already observed action $j$'s reward and it is indeed the maximum possible reward, i.e., $x_j=1$. As we are about to show, such a scenario is possible only when $x_1\in\{0,1\}$ and for every action $1<i<j$, it's reward has been observed, and $x_i=- 1$.
\item \textbf{An unknown reward}- action $j$ is yet to be explored, and every action $i<j$, has been explored and $x_i=-1$.
\end{enumerate}
\item \textbf{Exploration driven recommendation} Action $j$ has not been explored yet, every $i<j$ has been explored, and $x_i=-1$ for $i<j$, whereas $x_1=0$. This implies that the planner is in the exploration state of action $j$. 
\end{enumerate}

In Table \ref{tblExtj}, we are extending the algorithm described in Table \ref{RecTable} to all states that may result in a recommendation for action $j$ when using $\hat{\pi}$. We also added the gain of each state compared to action $1$ (i.e., $\E[u_t(j)-u_t(1)|S^{\vec{z}}_t,\sigma_t=j]$), as well as the probability that the planner is in these states in round $t$ (i.e., $\Pr[S_t^{\vec{z}}|\sigma_t=j]$).

A BIC exploration rate $q_j^t$ for our algorithm is still the maximum value that satisfy the same constraints as before (i.e., (\ref{qdefpart1})). 
\begin{theorem}
For k actions, given $q_2^2, \ldots,q_{t-1}^2,\ldots,q_{j-1}^{j-1},\ldots,q_{t-1}^{j-1}$, for $j>2$,
\begin{equation}
q_{t}^2 = 
     \begin{cases}
       0 &\quad t=1\\
       \min (\frac{2p_{1}^{-1}p_{2}^{1} }{1-2p_{2}^{1}},p_1^0) &\quad t=2\\
        \min(\frac{2p_{1}^{-1}p_{2}^{1}+p_{2}^{1}\sum_{m=2}^{t-1}{q^2_{m}}}{1-2p_{2}^{1}},
p_1^0-\sum_{m=2}^{t-1}{q_{m}^2})&\quad t>2\\
     \end{cases}
\end{equation}
\begin{equation}\label{qj}
q^j_{t} = 
     \begin{cases}
       0 &\quad t<j\\
       \min(\frac{2p_{j}^{1}\Pi_{i<j}p_{i}^{-1}}{1-2p_{j}^{1}},p_{j-1}^{- 1}q^{j-1}_{j-1}) &\quad t=j\\
       \min 
(\frac{2p_{j}^{1}\Pi_{i<j}p_{i}^{-1}+p_{j}^{1}\sum_{\tau=j}^{t-1}{q^j_{\tau}}}{1-2p_{j}^{1}} 
,p_{j-1}^{- 1}\sum_{\tau=j-1}^{t-1}{q^{j-1}_{\tau}}-\sum_{\tau=j}^{t-1}{q^j_{\tau}}) &\quad t>j\\
     \end{cases}
\end{equation}

In addition we show that 
\begin{equation}\label{inc}
q_{t+1}^{j+1}\leq  p_j^{-1}q_{t}^j
\end{equation}
\end{theorem}
\begin{proof}
The proof is done by induction over action $j$ and agent $t$.
The induction base and $q_t^2$'s part are provided by Theorem \ref{BigThm}. For the induction step we assume that (\ref{qj}) and (\ref{inc}) holds for any $t_0<t$ and $j_0<j$.

As for $q_t^j$, by using the induction hypothesis, we know that an exploitation driven recommendation $\sigma_t=j$ can only come from the first three exploitation states described in Table \ref{tblExtj}. We also know from it that an exploration driven recommendation can only come from the last state  in Table \ref{tblExtj}. For this we notice that $\Pr[S_t^{\vec{z}}]=\Pr[S_t^{\vec{z}},\sigma_t=j]$ for every $S_t^{\vec{z}}\in \Gamma_t^{j+}$ and that $\Pr[S_t^{\vec{z}},\sigma_t=j]=q_t^j$
Hence, the BIC constraint for any agent $t\geq j$ is
\[
\E[u_t(j)-u_t(1)|\sigma_t=j]\Pr[\sigma_t=j] =
2p_j^1\Pi_{i<j}p_{i}^{-1}+(2p_{j}^{1}-1)q^j
_{t}+p_{j}^{1}\sum_{\tau=j}^{t-1}{q^j
_{\tau}}\geq 0
\]
Notice that the first and third state in Table \ref{tblExtj} have the same value for 
\[\Pr_{\hat{\pi}}[S_t^{\vec{z}}|\sigma_t=j]\;\E[u_t(j)-u_t(1)|S^{\vec{z}}_t,\sigma_t=j]\]
Therefore we merged them.
In order for $q^j_t$ to be a valid, it must also satisfy
\[
q_t^j=\Pr[\sigma_t=j,S_t^{\vec{z}}\in \Gamma_t^{j-}]\leq \Pr[S_t^{\vec{z}}\in \Gamma_t^{j-}]
\]
And by substituting $S_t^{\vec{z}}\in \Gamma_t^{j-}$ with the probability for the last state in Table \ref{tblExtj} we get
\[
q_t^j\leq p_{j-1}^{- 1}\sum_{\tau=j-1}^{t-1}{q^{j-1}_{\tau}}-\sum_{\tau=j}^{t-1}{q^j_{\tau}}
\]
Therefore,
\[
q_t^j= \min 
(\frac{2p_{j}^{1}\Pi_{i<j}p_{i}^{-1}+p_{j}^{1}\sum_{\tau=j}^{t-1}{q^j_{\tau}}}{1-2p_{j}^{1}} 
,p_{j-1}^{- 1}\sum_{\tau=j-1}^{t-1}{q^{j-1}_{\tau}}-\sum_{\tau=j}^{t-1}{q^j_{\tau}})
\]
We can upper bound $q_t^j$ as follows
\[
q^j_{t}\leq \frac{2p_{j}^{1}\Pi_{i<j}p_{i}^{-1}+p_{j}^{1}\sum_{\tau=j}^{t-1}{q^j_{\tau}}}{1-2p_{j}^{1}} \leq
\frac{2p_{j-1}^{1}\Pi_{i<j}p_{i}^{-1}+p_{j-1}^{1}\sum_{\tau=j}^{t-1}{q^j_{\tau}}}{1-2p_{j-1}^{1}} 
\]
the second inequality is correct due to the assumption that $p_j^1<p_{j-1}^1$.
From the induction hypothesis, we know that $q_{\tau}^{j}< p_{j-1}^{-1}q_{\tau}^{j-1}$ for every $\tau\leq t-1$. Hence,
\[
q^j_{t}\leq 
\frac{2p_{j-1}^{1}\Pi_{i<j}p_{i}^{-1}+p_{j-1}^{1}\sum_{\tau=j}^{t-1}{q^j_{\tau}}}{1-2p_{j-1}^{1}} \leq 
\frac{2p_{j-1}^{1}\Pi_{i<j}p_{i}^{-1}+p_{j-1}^{- 1}p_{j-1}^{1}\sum_{\tau=j-1}^{t-2}{q^{j-1}_{\tau}}}{1-2p_{j-1}^{1}}=
p_{j-1}^{-1}A_{t-1}^{j-1}=p_{j-1}^{-1}q_{t-1}^{j-1}\;.
\]
By using the induction hypothesis again, we get $0<\sum_{\tau=j}^{t-1}{q^j_{\tau}}<p_j^{-1}\sum_{\tau=j-1}^{t-2}{q^{j-1}_{\tau}}$, thus
\[
p_{j-1}^{-1}q_{\tau}^{j-1}\leq p_{j-1}^{-1}q_{\tau}^{j-1}+ p_{j-1}^{-1}\sum_{\tau=j-1}^{t-2}{q^{j-1}_{\tau}}-\sum_{\tau=j}^{t-1}{q^j_{\tau}}=p_{j-1}^{- 1}\sum_{\tau=j-1}^{t-1}{q^{j-1}_{\tau}}-\sum_{\tau=j}^{t-1}{q^j_{\tau}}
\]
which completes the proof.
\end{proof}

\end{document}